\documentclass[journal]{IEEEtran}

\usepackage{amssymb}
\usepackage{amsmath,amssymb,amsfonts}
\usepackage{multirow}
\usepackage{graphicx}
\usepackage{textcomp}
\usepackage{amsthm}

\usepackage{setspace}
\usepackage{subfigure}

\usepackage{graphicx}
\usepackage{algorithm}
\usepackage{algorithmic}
\usepackage{amsmath}
\usepackage{color}
\usepackage{amsfonts}

\ifCLASSINFOpdf
\else
\fi

\newtheorem{thm}{Theorem}
\newtheorem{lem}{Lemma}

\newtheorem{defn}{Definition}

\newtheorem{rem}{Remark}

\begin{document}

\title{An Architecture for Distributed Energies Trading in Byzantine-Based Blockchain}

\author{Jianxiong Guo,
	Xingjian Ding,
	Weili Wu,~\IEEEmembership{Member,~IEEE}
	\thanks{J. Guo and W. Wu are with the Department
		of Computer Science, Erik Jonsson School of Engineering and Computer Science, Univerity of Texas at Dallas, Richardson, TX, 75080 USA; X. Ding is with the School of Information, Renmin University of China, Beijing, CHN
		
		E-mail: jianxiong.guo@utdallas.edu}
	\thanks{Manuscript received April 19, 2005; revised August 26, 2015.}}

\markboth{Journal of \LaTeX\ Class Files,~Vol.~14, No.~8, August~2015}%
{Shell \MakeLowercase{\textit{et al.}}: Bare Demo of IEEEtran.cls for IEEE Journals}

\maketitle

\begin{abstract}
	With the development of smart cities, not only are all corners of the city connected to each other, but also connected from city to city. They form a large distributed network together, which can facilitate the integration of distributed energy station (DES) and corresponding smart aggregators. Nevertheless, because of potential security and privacy protection arisen from trustless energies trading, how to make such energies trading goes smoothly is a tricky challenge. In this paper, we propose a blockchain-based multiple energies trading (B-MET) system for secure and efficient energies trading by executing a smart contract we design. Because energies trading requires the blockchain in B-MET system to have high throughput and low latency, we design a new byzantine-based consensus mechanism (BCM) based on node's credit to improve efficiency for the consortium blockchain under the B-MET system. Then, we take combined heat and power (CHP) system as a typical example that provides distributed energies. We quantify their utilities, and model the interactions between aggregators and DESs in a smart city by a novel multi-leader multi-follower Stackelberg game. It is analyzed and solved by reaching Nash equilibrium between aggregators, which reflects the competition between aggregators to purchase energies from DESs. In the end, we conduct plenty of numerical simulations to evaluate and verify our proposed model and algorithms, which demonstrate their correctness and efficiency completely.
\end{abstract}

\begin{IEEEkeywords}
	Distributed energies trading, Smart city, Consortium blockchain, Byzantine consensus, Stackelberg game.
\end{IEEEkeywords}

\IEEEpeerreviewmaketitle

\section{Introduction}
\IEEEPARstart{T}{he} deployment of distributed energy stations (DESs) based on the internet built by the development of smart cities has been a hot topic because of its great potential to reduce the consumption of fossil fuels and curb greenhouse gas emission \cite{georgilakis2013optimal}. Consider a smart community equipped with a DES, it is used to supply residents in this community with multiple energies, such as electricity and heat. DES existing in the community can reduce residents' dependence on the centralized supply of energies, such as electricity from power grid and heat from heat station, thus save resources and reduce the cost of using energies. Moreover, it can sell surplus electricity and heat to the aggregators of power grid and heat station for making revenue. DESs can trade their surplus energies with aggregators that are responsible for collecting energies from their communities in a peer-to-peer (P2P) manner, thereby the multiple energies trading problem discussed in this paper is formulated.

Traditional P2P energies trading is performed on the centralized energy management platform, however, such a mechanism has many drawbacks. Traders often worry that their payment security and privacy protection when trading in an untrusted and opaque third centralized platform. This intermediary needs to verify and manage transactions between aggregators and DESs. If troubled by some damages such as single point of failure, it will lead to privacy leakage and transaction loss \cite{aitzhan2016security}. Thus, it is urgent to create a secure energies trading system to guarantee trading among the distributed internet of energy can be executed effectively. It encourages the DESs to sell their energies to aggregators without worry, which promotes the rational use of energies.

Blockchain is a public and distributed database that is designed to store verified transactions among all valid participants without a trusted intermediary. Here, a new transaction is required to be validated by a group of authorized participants, and then it can be added into the blockchain in a permanent and tamper-resistant manner. It can be used to construct a secure and reliable energies trading system because of its decentralization, security, and anonymity \cite{li2018crowdbc} \cite{jiao2019auction}. Consider a smart city, it consists of a number of communities, each of which is equipped with a DES. There are two aggregators, electricity aggregator (EA) and heat aggregator  (HA), trading with DESs in this city. The aggregators of different cities are interconnected to form a wide area network. Based on that, we propose a blockchain-based multiple energies trading (B-MET) system, where all aggregators are authorized participants required to store the blockchain and complete the consensus process. Thus this is a consortium blockchain as well, which is a little different from the classical public blockchain used in Bitcoin and Ethereum. Here, consortium blockchain is more convenient and flexible to achieve trading functions.

Based on such an architecture, we design a smart contract that ensures energies trading to be performed automatically when the trading conditions are satisfied. However, the proof-based consensus mechanism such as proof-of-work that is adopted by the most of blockchain applications is not suitable to our consortium blockchain in B-MET system because of its high latency, low throughput, and demanding computing power requirement. To finish the task of energies trading, it needs low latency and high throughput consensus mechanism. Thereby we design a new byzantine-based consensus mechanism (BCM) based on node's credit, which reflects the performance of this node in the previous experience of participating in consensus. After each round of consensus, each node's credit should be updated according to its voting result. If its voting is consistent with the result of consensus, its credit will be increased; otherwise will be decreased. Their credits affect directly their probabilities of being chosen as the leader and voting weight in the next round. This not only motivates participants to make the right decision, but also speeds up the consensus process.

In the aforementioned contract, there is an interaction between aggregator and DES before initiate a new energies trading, where the aggregator offers a unit price to purchase a kind of energy from DES, then DES decides the amount of energy they are willing to sell. In this paper, we take combined heat and power (CHP) system as an instance of DES, and aggregators are EA and HA. In a smart city, for each DES in this city, its utility consists of two parts: one is to serve the residents living in the community for satisfying their Daily consumption, and the other is sold to the aggregators for gaining revenues. For the aggregators in this city, their gains come from buying energies from DESs at a lower price and selling them at the retail price. To motivate the DESs to sell more energies, the aggregators should offer a higher price to them, but doing so raises costs and may result in lower overall profits. Since the multilevel decision-making processes between aggregators and DESs in a city, we formulate a novel multi-leader multi-follower (MLMF) Stackelberg game to model this bargain between them. Here the aggregators are leaders and DESs are followers. Their goals are to maximize their utilities or profits respectively. This MLMF Stackelberg game is analyzed and solved thoroughly in this paper, and we prove the Nash equilibrium (NE) among aggregators exists and is unique. Because the DESs are always able to respond aggregators with the optimal strategy according to their offered prices, the Stackelberg equilibrium (SE) exists and is unique as well. We propose a distributed algorithm that is guaranteed to reach the unique SE by limited information interactions. Finally, we conduct extensive numerical simulations to test the B-MET system, verify the correctness of our proposed utility functions and feasibility of our proposed algorithm.

The rest of this paper is organized as follows: Sec. \uppercase\expandafter{\romannumeral2} discusses the-state-of-art work. Sec. \uppercase\expandafter{\romannumeral3} introduces the architecture of B-MET system, describes CHP system, and defines utility functions. Sec. \uppercase\expandafter{\romannumeral4} presents smart contract and byzantine-based blockchain. Sec. \uppercase\expandafter{\romannumeral5} introduces the Stackelberg game and discuess the solving process. Sec. \uppercase\expandafter{\romannumeral6} conducts numerical simulations. Sec. \uppercase\expandafter{\romannumeral7} is conclusion.

\section{Related Work}
	Distributed energy systems have been applied widely in many different forms, such as DES \cite{wu2006combined} and vehicle-to-grid \cite{su2018secure} \cite{zhou2019secure}, to curb greenhouse gas and save cost. Integrating DESs into a smart grid \cite{vasquez2010hierarchical} has attracted more and more researchers to participate in recently. This rouse the problems of energy management and energy trading problem. Cecati \textit{et al.} \cite{cecati2011smart} exploited DES to make the cost of power delivery	minimized by use of an efficient smart grid management system. Georgilakis \textit{et al.} \cite{georgilakis2013optimal} summarized the optimally distributed generation placement problem systematically, classified and analyzed current and future research about it. Zhang \textit{et al.} \cite{zhang2013efficient} considered microgrid as a local energy supplier for domestic buildings by utilizing DES, and studied optimal scheduling of energy consumption through mixed-integer programming. However, they only focused on electricity trading between grid and DESs, in this paper we consider multiple energies trading due to the diversity of energy forms.
	
	In P2P energy trading, blockchain technology has been introduced to address transaction security issues. Kang \textit{et al.} \cite{kang2017enabling} put forward a localized P2P electricity trading pattern based on consortium blockchain among plug-in hybrid electric vehicles. Li \textit{et al.} \cite{li2017consortium} proposed a P2P energy trading architecture based on consortium blockchain for the industrial internet of things relied on a credit-based payment scheme. zhou \textit{et al.} \cite{zhou2019secure} considered the scenario of vehicle-to-grid, and developed a secure energy trading mechanism based on consortium blockchain. Guo \textit{et al.} \cite{guo2020combined} studied a blockchain-based energy management system that guarantees secure electricity trading between grid and DESs. However, they lose sight of low throughput and high latency in their proof-based consensus process, in this paper we try to address it by proposing a new byzantine-based consensus mechanism.
	
	Stackelberg game is an effective tool to model the interactions in energies trading. Maharjan \textit{et al.} \cite{maharjan2013dependable} studied the demand response management by establishing a Stackelberg game between multiple utility companies and customers to maximized their utilities respectively. Bu \textit{et al.} \cite{bu2013game} proposed a four-stage Stackelberg game to consider a real-time pricing problem for the electricity retailer in the demand-side management. Yao \textit{et al.} \cite{yao2019resource} modeled the interactions between cloud server and mines by Stackelberg game, and solved it by multiagent reinforcement learning algorithm. Chen \textit{et al.} \cite{chen2020stackelberg} proposed a Stackelberg game-based framework to simulate the multiple resources allocation between cloud server and end users, and found an equilibrium solution by a backward induction process. However, most of these models have only one leader, in this paper the interactions between aggregators and DESs are MLMF, more complex and realistic.  

\section{System Architecture}
	Consider a smart city, it consists of a number of disjoint smart communities, each of which is equipped with a distributed energy station (DES) responsible for supplying multiple energies, such as electricity and heat, to these residents living in this community. In this city, there are several aggregators, which represent different companies respectively, collecting different kinds of energies from all DESs appertained to this city. The architecture of blockchain-based multiple energies trading (B-MET) system is shown in Fig. \ref{fig1}. In the B-MET system, given a smart city ${\rm S}_i$, the entities in this smart city can be shown as follows:
	
	\begin{figure}[!t]
		\centering
		\includegraphics[width=\linewidth]{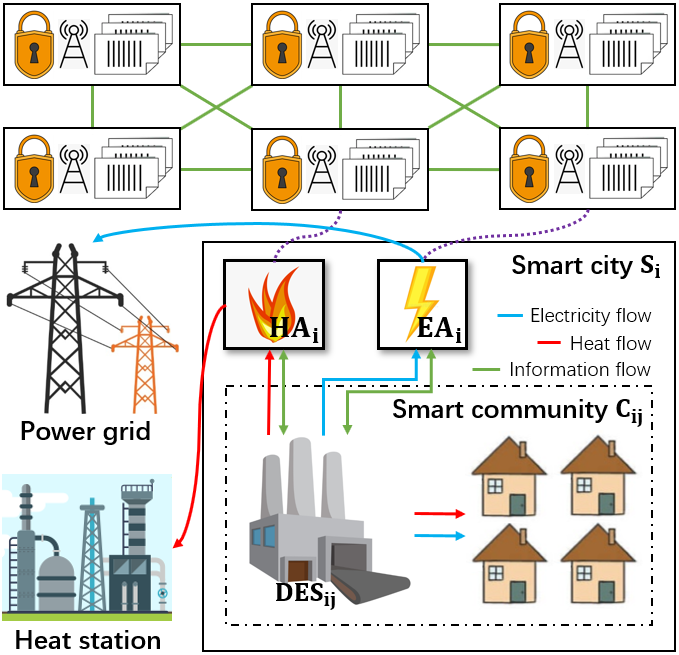}
		\caption{The architecture of blockchain-based multiple energies trading system.}
		\label{fig1}
	\end{figure}
	
	\textit{1) Aggregators:} There are two aggregators, electricity aggregator (${\rm EA}_i$) and heat aggregator (${\rm HA}_i$), associated with this smart city ${\rm S}_i$. The ${\rm EA}_i$ (resp. ${\rm HA}_i$) is delegated by power grid (resp. heat station) as a monopoly of the energy market. They purchase electric energy (resp. heat energy) generated by DESs in those communities that belong to this smart city.
	
	\textit{2) DESs:} The city ${\rm S}_i$ can be partitioned into an uncertain number of disjoint smart communities, denoted by set $\{{\rm C}_{i1},{\rm C}_{i2},\cdots,{\rm C}_{ij},\cdots\}$. In community ${\rm C}_{ij}$, there is a distributed energy station ${\rm DES}_{ij}$ supplying electricity and heat to the residents living in this community. Besides, ${\rm DES}_{ij}$ is able to sell surplus electric energy (resp. heat energy) to the corresponding ${\rm EA}_i$ (resp. ${\rm HA}_i$) in order to make revenues.
	
	\textit{3) Smart meters:} It is a built-in component installed in each aggregator that monitors the energy flow transferred by each DES in this city in real-time, and decide whether the transaction has been accomplished.

	Then, consider a larger ecosystem, such as a country, it is composed of a number of smart cities. This ecosystem $\mathbb{S}$ can be denoted by $\mathbb{S}=\{{\rm S}_1,{\rm S}_2,\cdots,{\rm S}_i,\cdots\}$. Here, each ${\rm S}_i\in\mathbb{S}$ is a smart city in this ecosystem, and ${\rm S}_i=\{\{{\rm EA}_i,{\rm HA}_i\},\{{\rm C}_{i1},{\rm C}_{i2},\cdots,{\rm C}_{ij},\cdots\}\}$. For convenience, the notation ${\rm DES}_{ij}$ can be considered equivalent to ${\rm C}_{ij}$. Our B-MET system is established on such an ecosystem, in which all aggregators, including EAs and HAs, are interconnected each other to form a peer-to-peer (P2P) network called ``blockchain network'', shown in the upper half of Fig. \ref{fig1}. In order to support secure energy trading between aggregators and DESs, we adopt consortium blockchain to construct our B-MET. In traditional blockchain, the consensus process is carried out by all participants. But the blockchain in the B-MET takes all aggregators in the ecosystem as authorized participants, and they are charged with storing the whole blockchain and performing the consensus process. Each aggregator manages and records those transactions between it and DESs in its city. The transactions are packaged into blocks and added into blockchain when the consensus among aggregators is reached, thus stored in all aggregators permanently.
	
\subsection{Combined Heat and Power System}
	Here, the aforementioned DES is implemented by the combined heat and power (CHP) system. The CHP system consumes natural gas to generate electricity and heat that serve its community or sell to the aggregators of its corresponding city, shown in Fig. \ref{fig2}. The gas is fed into the gas turbine (GT) which will generate electricity $E_g$ and emit high-temperature waste heat $Q_w$. The heat $Q_w$ can be recovered by heat recovery system that can generate heat $Q_r$. Here, $E_{use}$ (resp. $Q_{use}$) is used to supply electricity (resp. heat) to community, and $E_{exc}$ (resp. $Q_{exc}$) is sold to EA (resp. HA).
	
	\begin{figure}[!t]
		\centering
		\includegraphics[width=\linewidth]{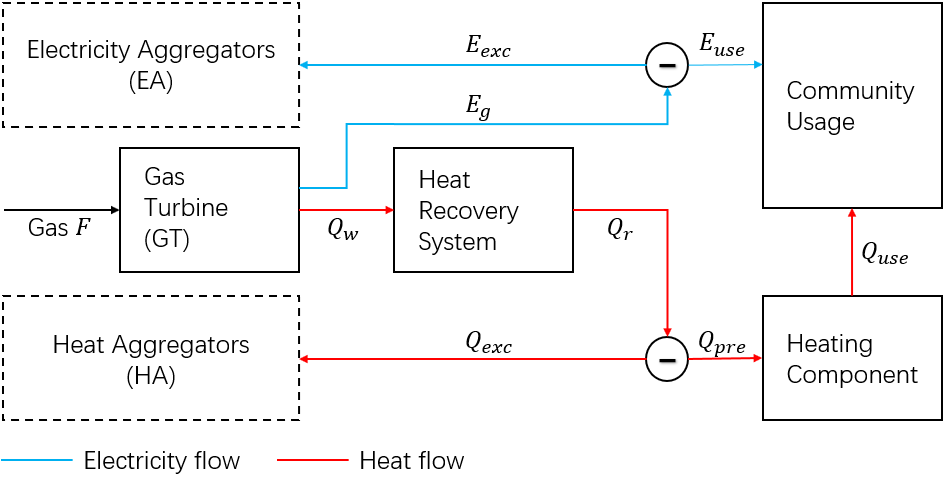}
		\caption{The structure of combined heat and power (CHP) system.}
		\label{fig2}
	\end{figure}

	Shown as Fig. \ref{fig2}, the total electricity generated by GT is $E_{g}=E_{use}+E_{exc}$. Measured in days, the units of quantities denoted by $E$ and $Q$ are $({\rm J}/{\rm day})$. The gas consumption per day $F$ $({\rm m}^3/{\rm day})$ can be defined as
	\begin{equation}
		F=E_g/(q\cdot\eta_g)=Q_w/(q\cdot(1-\eta_g))
	\end{equation}
	where $q$ $({\rm J}/{\rm m}^3)$ is the calorific value of natural gas, thereby the total energy generated by $F$ is $q\cdot F$ definitely. The $\eta_g$ is the electric conversion of GT, percentage energy that transferred to electricity. Given a specific GT, its electric conversion can be considered as a constant. Besides, let $\eta_r$ be the thermal efficiency of heat recovery system, and $\eta_h=1$ be the thermal efficiency of heat component. We have $Q_r=Q_w\cdot\eta_r=Q_{pre}+Q_{exc}$ and $Q_{use}=Q_{pre}$ respectively.
	
	As mentioned above, for a given CHP system, the electricity (resp. heat) generated by it can be divided into two parts: one is used to serve local residents, another is sold to EA (resp. HA). Thus, we define two dispatching factor $\alpha,\beta\in[0,1]$ for this CHP, where $\alpha=E_{use}/E_g$ is the electric dispatching factor and $\beta=Q_{pre}/Q_r$ is the heat dispatching factor. This DES needs to buy natural gas from the gas company. The company is for profit, thus it is valid to assume the gas company always supply enough gas that is able to meet the DES's requirement. Given a smart city ${\rm S}_i$ and a smart community ${\rm C}_{ij}\in{\rm S}_i$, the energy relationships in the CHP system of ${\rm C}_{ij}$ has been obtained, that is
	\begin{flalign}
	E_{use}^{ij}&=\alpha^{ij}\cdot\eta_g\cdot(qF^{ij})\\
	E_{exc}^{ij}&=(1-\alpha^{ij})\cdot\eta_g\cdot(qF^{ij})\\
	Q_{use}^{ij}&=\beta^{ij}\cdot(1-\eta_g)\cdot\eta_r\cdot(qF^{ij})\\
	Q_{exc}^{ij}&=(1-\beta^{ij})\cdot(1-\eta_g)\cdot\eta_r\cdot(qF^{ij})
	\end{flalign}

	In this model, we assume all the CHP equipments in this ecosystem $\mathbb{S}$ has the same efficiency parameters $\eta_g$ and $\eta_r$. Each ${\rm CHP}_{ij}$ can determine the amount of electricity (resp. heat) that can be sold to ${\rm EA}_i$ (resp. ${\rm HA}_i$) by adjusting its dispatching factor $\alpha^{ij}$ (resp. $\beta^{ij}$) automonously. For example, when $\alpha^{ij}$ is one, it means that ${\rm CHP}_{ij}$ will not sell any electricity to ${\rm EA}_i$ for making revenue.
	
\subsection{Utility Functions}
	Consider a smart city ${\rm S}_i$, the ${\rm EA}_i$ (resp. ${\rm HA}_i$) offers a unit price $p_e^i$ (resp. $p_h^i$) to collect surplus electricity (resp. heat) generated by ${\rm DES}_{ij}\in{\rm S}_i$, where the units of $p_e^i$ and $p_h^i$ are ${\rm coin}/{\rm J}$. For each ${\rm DES}_{ij}\in{\rm S}_i$, it is a risk-averse agent in the energy market. If ${\rm DES}_{ij}$ chooses dispatching factor $\alpha^{ij}$, $\beta^{ij}$, and consume natural gas $F^{ij}$, that is
	\begin{flalign}
		U^{ij}(\alpha^{ij},\beta^{ij},F^{ij})&=W_e^{ij}(E_{use}^{ij})+W_h^{ij}(Q_{use}^{ij})\nonumber\\
		&+p_e^i\cdot E_{exc}^{ij}+p_h^i\cdot Q_{exc}^{ij}-c_f\cdot F^{ij}
	\end{flalign}
	where $W_e^{ij}$ (resp. $W_h^{ij}$) is the satisfaction function of community ${\rm C}_{ij}$ that provides electricity (resp. heat) to satisfy the usage of local residents in this community, and $E_{use}^{ij}$, $E_{exc}^{ij}$, $Q_{use}^{ij}$, and $Q_{exc}^{ij}$ are defined from (2) to (5). Here, $c_f$ $({\rm coin}/{\rm m}^3)$ is the unit cost of natural gas.
	
	From our simplified CHP model, we denote the cost of electricity (resp. heat) produced from ${\rm DES}_{ij}$ by $c_e$ (resp. $c_h$). Then, the cost $({\rm coin}/{\rm J})$ of electricity and heat can be quantified, that is $c_e=c_f/q$ and $c_h=c_f/(q\cdot\eta_r)$. Thus, we have $c_f\cdot F^{ij}=\eta_g\cdot c_e\cdot(qF^{ij})+(1-\eta_g)\cdot c_h\cdot\eta_r\cdot(qF^{ij})$. If the price $p_e^i$ (resp. $p_h^i$) offered by ${\rm EA}_i$ (resp. ${\rm HA}_i$) is less than the cost $c_e$ (resp. $c_h$), this ${\rm DES}_{ij}$ will not sell any electricity (resp. heat) to them. It will reduce the gas intake $F^{ij}$ such that only meet its local requirement. Like this, there is no energies trading between aggregators and DESs, and obviously, it is not what we want to see. Thus, it is reasonable to consider the prices offered by aggregators satisfy $p_e^i\geq c_e$ and $p_h^i\geq c_h$. At this time, for each ${\rm DES}_{ij}\in{\rm S}_i$, it will produce electricity and heat as much as possible, because of the fact that it is always profitable to sell them to the aggregators. For maximizing its utility, each CHP system will run at full capacity. Here, for each ${\rm CHP}_{ij}$, we define its maximum production capacity (maximum gas consumption) per day as $F_m^{ij}$. Therefore, the utility $U^{ij}(\alpha^{ij},\beta^{ij},F_m^{ij})$ can be denoted by $U^{ij}(\alpha^{ij},\beta^{ij})$, because $F_m^{ij}$ is considered as a constant.
	
	\begin{rem}
		After here, we denote $X^{ij}=\eta_g\cdot(qF^{ij}_m)$ and $Y^{ij}=(1-\eta_g)\cdot\eta_r\cdot(qF^{ij}_m)$ for convenience.
	\end{rem}
	
	Based on \cite{maharjan2013dependable} \cite{tushar2014three} \cite{su2018secure}, the natural logarithmic functions were adopted extensively in characterizing the satisfaction of comsuming energy. That is
	\begin{flalign}
		W_e^{ij}(E_{use}^{ij})&=k_e^{ij}\cdot\ln(1+b_e^{ij}\cdot E_{use}^{ij})\\
		W_h^{ij}(Q_{use}^{ij})&=k_h^{ij}\cdot\ln(1+b_h^{ij}\cdot Q_{use}^{ij})
	\end{flalign}
	where $k_e^{ij}$ (resp. $k_h^{ij}$) is a non-negative satisfaction coefficient for electricity (resp. heat) in community ${\rm C}_{ij}$, and $b_e^{ij}$ (resp. $b_h^{ij}$) is a non-negative adaption coefficient electricity (resp. heat) in this community as well. The adaption coefficients were proposed in \cite{guo2020combined} first, which aimed to control the variation range of the term $\ln(1+\cdot)$, avoid it growing infinitely. Generally, we let $\ln(1+b_e^{ij}\cdot E_{use}^{ij})=1$ (resp. $\ln(1+b_h^{ij}\cdot Q_{use}^{ij})=1$) when we choose $\alpha^{ij}=1$ (resp. $\beta^{ij}=1$) by setting a valid adaption coefficient $b_e^{ij}$ (resp. $b_h^{ij}$) \cite{guo2020combined}. Base on that, thereby we can formulate $b_e^{ij}$ and $b_h^{ij}$ as follows:
	\begin{equation}
		b_e^{ij}=(1/X^{ij})\cdot(e-1)\text{; }b_h^{ij}=(1/Y^{ij})\cdot(e-1)
	\end{equation}
	
	For aggregators in this city, power grid and heat station are the retailers for electricity and heat, however they do not have pricing power, because the retail prices of electricity and heat subject to government's regulation. Hence, we define a retail price $r_e$ (resp. $r_h$) of electricity (resp. heat). As a selfish participant, it requires that $p_e^i\in[c_e,r_e]$ and $p_h^i\in[c_h,r_h]$. From (6), if $p_e^i$ (resp. $p_h^i$) offered by ${\rm EA}_i$ (resp. ${\rm HA}_i$) is too low, each ${\rm DES}_{ij}\in{\rm S}_i$ will respond with raising its dispatching factor $\alpha^{ij}$ (resp. $\beta^{ij}$), and sell less energy to aggregators. If the aggregators offer a high price to purchase energy, their profitable spaces are reduced even if DESs are willing to sell more energies to them. Both of these cases will cause aggregators' profit to be cut down. Therefore, it is important for aggregators to offer a optimal price such that not only encourage DESs to sell more energies, but also ensure sufficient profitability. If ${\rm EA}_i$ (resp. ${\rm HA}_i)$) offers a price $p_e^i$ (resp. $p_h^i$), its profit function can be defined as
	\begin{flalign}
		V_e^{i}(p_e^i,p_h^i)&=(r_e-p_e^i)\cdot\sum\nolimits_{{\rm C}_{ij}\in{\rm S}_i}E_{exc}^{ij}\\
		V_h^{i}(p_h^i,p_e^i)&=(r_h-p_h^i)\cdot\sum\nolimits_{{\rm C}_{ij}\in{\rm S}_i}Q_{exc}^{ij}
	\end{flalign}
	where $V_e^{i}$ (resp. $V_h^{i}$) is the profit function of the ${\rm EA}_i$ (resp. ${\rm HA}_i$) that collects electricity (resp. heat) from DESs in its city, and $E_{exc}^{ij}$ and $Q_{exc}^{ij}$ are defined in (3) and (5).
	
\section{Byzantine-Based Blockchain}
In this section, we will introduce a smart contract used to perform energies trading, and design a novel byzantine-based consensus mechanism based on the B-MET system.

\subsection{Smart Contract}
	A smart contract is a collection of programmable digital agreement that every participant commit to comply. Under our blockchain-based energies trading ecosystem, a transaction can only happen between aggregators and DESs in the same city. Thereby, consider a city ${\rm S}_i$, a smart contract can be decided together by its participants, which consist of an aggregator $k\in\{{\rm EA}_i,{\rm HA}_i\}$ and a ${\rm DES}_{ij}\in{\rm S}_i$. We denote such a smart contract by $Contract(k,{\rm DES}_{ij},STime)$. Between anonymous and untrusted entities in a city, the smart contract is able to execute credible transactions without third institutions. Then, the procedure of its smart contract $Contract(k,{\rm DES}_{ij},STime)$ is presented as follows:
	
	\textit{1) System initialization: }At the beginning, each ${\rm DES}_{ij}\in{\rm S_i}$ needs to acquire a unique identification $ID_{ij}$ by registering in the designated institution authorized by government. It will be assigned with its public/private key pair $(PK_{ij},SK_{ij})$ and a $Account_{ij}$. That is
	\begin{equation*}
		\{ID_{ij},PK_{ij},SK_{ij},Account_{ij}\}\leftarrow\text{register}({\rm DES}_{ij})
	\end{equation*}
	where each account is associated with its wallet address and balance, $Account_{ij}\leftarrow\{ Address_{ij},Balance_{ij}\}$. Then, for each aggregator $k\in\{{\rm EA}_i,{\rm HA}_i\}$ in this city, it has those necessary information $\{ID_{k},PK_{k},SK_{k},Account_{k}\}$ as well. But there is a credit value $Credit_k$ that represent the reputation of aggregator $k$, thereby we have $ Account_{k}\leftarrow\{Address_k,Balance_k,Credit_k\}$. Here, technologies of asymmetric encryption are usually adopted by current blockchain system for the sake of security, privacy, and data integrity. Given a massage $msg$ encrypted by ${\rm DES}_{ij}$, we have $Hash(msg)=PK_{ij}(SK_{ij}(Hash(msg)))$, where the unforgeability and integrity is guaranteed.
	
	\textit{2) Creation: }An aggregator $k\in\{{\rm EA}_i,{\rm HA}_i\}$ offers a price $p_k$ to buy energy from communities in its city, then ${\rm DES}_{ij}$ responds it with the amount of energy $x\in\{E_{exc}^{ij},Q_{exc}^{ij}\}$ that can be sold to $k$. Like this, a new smart contract $ Contract( k,{\rm DES}_{ij},STime)$ is generated by signing with their private key respectively. Then, this contract will be broadcasted to all authorized participants (aggregators) in the ecosystem $\mathbb{S}$. After reaching a consensus, this smart contract will be deployed and executed automatically. Each smart contract between aggregator and DES, $Contract(k,{\rm DES}_{ij},STime)$, is associated with several variables, which include account information $(Account_{k},Account_{ij})$, offered price $p_k$, amount of energy $x$, expected transaction time $TransTime$, and timestamp $STime$. To guarantee this contract can be executed successfully, it needs to verify whether aggregator $k$ has sufficient balance such that ${\rm Balance}_k\geq p_k\cdot x$ and whether ${\rm DES}_{ij}$ has enough production capacity to supply $x$ amount of corresponding energy on time.
	
	\textit{3) Execution: } The $Contract(k,{\rm DES}_{ij},STime)$ will be executed if current time $t\geq{\rm TransTime}$ after reaching a consensus among aggregators in blockchain network. From now on, it begins to trade energy and finish payment. The smart meter in aggregator $k$ verifies whether the amount of energy has been transported to the designated location. Then, fed this result from smart meter into the smart contract, if yes, it will execute the payment process automatically, that is
	\begin{equation*}
		(k,Balance_k-p_k\cdot x)\text{; }({\rm DES}_{ij},Balance_{ij}+p_k\cdot x)
	\end{equation*}
	Here, we design a mechanism that the balance is permitted to be negative. At the moment of payment, the smart contract will complete payment as usual if the $k$'s balance is not enough to pay. In this way, the $k$'s balance will become negative. Then, any contract that aggregator $k$ participants in will not be executed until its balance back to be positive.
	
	Generally speaking, the energies trading between aggregators and DESs can be summarized as follows: In a smart city, a DES begins a smart contract with an aggregator by responding it according to its offered price. This contract needs to be verified by the consensus process in blockchain network. Then, it will execute the predefined procedure automatically once the trading conditions are met, which achieves the digital currency and energy exchange specified by contract between participants in a secure manner.
	
\subsection{Byzantine-based Consensus mechanism}
	As mentioned early, a consensus process is necessary to be performed so as to ensure the consistency of blockchain stored in every authorized node. Castro \textit{et al.} \cite{castro1999practical} proposed a practical byzantine fault tolerance (PBFT) algorithm, which has been used in consortium blockchain system widely. Based on it and combined with the characteristics of energies trading, we design a new byzantine-based consensus mechanism (BCM). Our consensus process is performed among all aggregators in the ecosystem $\mathbb{S}$, and each of them has a local transaction pool (LTP) to store all transactions it receives. The consensus process is executed round by round, and the time interval of block generation is given by $\Delta T$. There are three main stages, shown in Fig. \ref{fig3}, that is pre-prepare, prepare, and commit.
	
	\begin{figure}[!t]
		\centering
		\includegraphics[width=\linewidth]{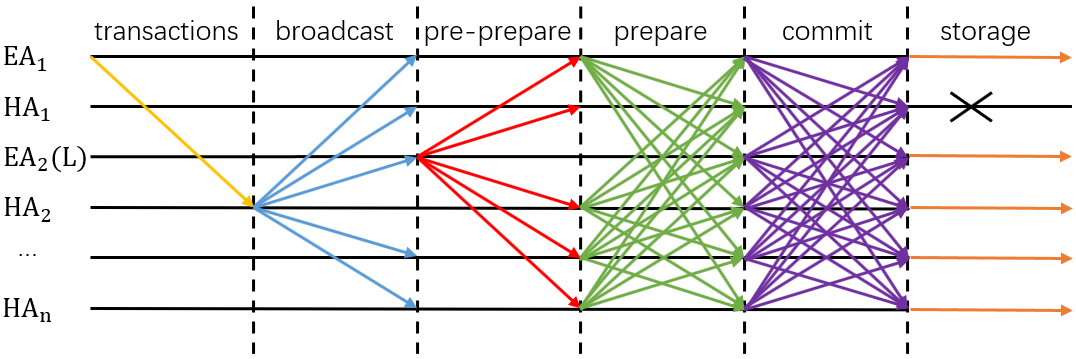}
		\caption{The consensus process of BCM.}
		\label{fig3}
	\end{figure}
	
	First, let us introduce the credit model, which will be used in leader election and to decide whether to reach a consensus. Let $M=\{{\rm EA}_1,{\rm HA}_1,{\rm EA}_2,{\rm HA}_2,\cdots\}$ be the collection of consensus nodes. We have known that there is an attribute $Credit_k\in[0,1]$ for each $k\in M$, where a larger $Credit_k$ implies node $k$ is more trustworthy. We denote by $Credit_k(i)$ the credit of node $k$ after finishing the $i$-th round consensus. Then, we can define $Credit_k(i+1)$ according to the result of consensus in the $(i+1)$-th round, where this result is whether to add the leader's block into the blockchain. That is: (1) when $k$ is the leader, we have $Credit_k(i+1)=\min\{1,Credit_k(i)+\Delta_1\}$ if its block is accepted to be added into the blockchain, else $Credit_k(i+1)=\max\{0,Credit_k(i)-\Delta_1\}$ if its block is rejected; and (2) when $k$ is not the leader, we have $Credit_k(i+1)=\min\{1,Credit_k(i)+\Delta_2\}$ if its decision is consistent with consensus result; else $Credit_k(i+1)=\max\{0,Credit_k(i)-\Delta_2\}$ if it disagrees with the majority. We usually give $\Delta_1>\Delta_2>0$ and initialize the credit of each consensus node as $credit_k(0)=0.5$. Consider entering the $(i+1)$-th round consensus, the detailed process of BCM is represented as follows:
	
	\textit{1) Leader election: }The first step in this round is to select a leader from all consensus nodes. This leader election is based on node's credit. Generally speaking, the better the credit value of a node, the more likely it is to be elected as the leader. Thus, the result of leader selection is unpredictable. For a node $k\in M$, the probability that it is elected as the leader of the $(i+1)$-th round consensus is  $\Pr[L(i+1)=k]$,
	\begin{equation}
		\Pr[L(i+1)=k]=\frac{Credit_k(i)}{\sum_{j\in M}Credit_j(i)}
	\end{equation}
	where $L(i+1)$ represents the leader of the $(i+1)$-th round consensus. Obviously, there is no chance to select a node whose credit is zero as the leader.
	
	\textit{2) Broadcast: }Each aggregator in $M$ broadcasts all transactions which happen in current $\Delta t$ and co-signed with a DES in its city to the blockchain network. All the consensus nodes will verify whether their received transactions are valid. Those valid transactions will be stored in their LTP, and invalid transactions will be discarded.
	
	\textit{3) Pre-prepare: }After all non-leader consensus nodes in $M\backslash L(i+1)$ have completed above verification process for received transactions, the leader will package those selected valid transactions in its LTP into a block $B_L$. Then, the leader signs this block and broadcasts pre-prepare message $SK_{L}(SK_{L}(B_L),pre$-$prepare)$ to the blockchain network.
	
	\textit{4) Prepare: }For each non-leader node $k\in M\backslash L(i+1)$, it will check the identity of leader and verify the pre-prepare message from the leader. The block verification needs to confirm the pointer to the previous block, mercle root is correct and compare the transactions in $B_L$ with the corresponding transactions in its LTP. If node $k$ believes $B_L$ is valid, it broadcasts this prepare message $SK_k(SK_{L}(B_L), prepare)$ to the blockchain network. All consensus nodes must make decisions in this step, whether to agree or disagree with adding block $B_L$ into the blockchain. Then for each node $k\in M$, it gathers all prepare massages from other consensus node, checks their identities and counts the weighted sum of received prepare messages. Let $A_k(i+1)\subseteq M$ be the set of nodes from which node $k$ receives prepare messages, including itself. If satisfying the following inequality
	\begin{equation}
		\sum_{a\in A_k(i+1)}\Pr[a]\geq\left(2\left\lfloor\frac{|M|-1}{3}\right\rfloor+1\right)\frac{1}{|M|}
	\end{equation}
	where $\Pr[a]={Credit_a(i)}/{\sum_{j\in M}Credit_j(i)}$, we say node $k$ will accept block $B_L$ and broadcast commit message $SK_k(SK_{L}(B_L),commit)$ to the blockchain network.
	
	\textit{5) Commit: }After sending their commit messages, they should waiting commit messages from other consensus node. For each node $k\in M$, its consensus process is completed until it recieve sufficient commit messages such that $\sum_{a\in B_k(i+1)}\Pr[a]\geq(2\lfloor(|M|-1)/{3}\rfloor+1)/|M|$, where $B_k(i+1)\subseteq M$ is the set of nodes from which node $k$ receives commit message, including itself.
	
	\textit{5) Add a block and update credits: }If a consensus node accepts the new block $B_L$, it will be appended into the blockchain in a linear and chronological order, which includes a pointer to the previous block. Any failure occurs in these three stages will terminate the consensus of current round (do not add the new block). Besides, before finishing this round, we need to update the credits of all the consensus nodes according to the credit model. In next round, the consensus process will perform based on their new credits.
	
	Failures that terminate the current round and do not add a new block mainly include the following reasons: (1) the leader sends invalid block or do not send its packaged block before the deadline; and (2) too many malicious nodes do not breadcast prepare messages even though this block is valid. Shown as node ${\rm HA}_1$ in Fig. \ref{fig3}, it is a faulty node. The credit of those nodes that make mistakes in this consensus process will be reduced. Let $f$ be the number of malicious nodes. According to \cite{castro1999practical}, supposing  $f\leq\lfloor(|M|-1)/3\rfloor$, the faults can be tolerated by the consensus system with $|M|$ nodes. In our BCM, each consensus node's credit is initialized as a constant, thereby $2\lfloor(|M|-1)/3\rfloor+1$ good nodes can make sure that (13) is satisfied. As the consensus process performs more and more times, the good nodes' credit increase but malicious nodes' credit decrease gradually. Therefore, the credits in system will be more accumulative in good nodes. According to (13), the number of prepare message and commit message required to reach a consensus declines, which helps reduce latency and improve throughput. In summary, secure and traceable energies trading and digital currency exchange can be guaranteed by our proposed B-MET system.

\section{Multiple Energies Trading: A Stackelberg Approach}
	A non-cooperative Stackelberg game generally refers to the multilevel decision making processes of a number of independent decision-makers in response to the decision taken by the leading player of the game \cite{doi:10.1137/1.9781611971132}. In this section, we put forward a multi-leader multi-follower (MLMF) Stackelberg game to model the interactions in above smart contract between aggregators and DESs. Consider a smart city ${\rm S}_i$, the MLMF Stackelberg game $\mathbb{G}$ can be defined as
	\begin{equation}
		\mathbb{G}=\left\{{\rm S}_i,\mathbb{P},\mathbb{D},\{V_e^i,V_h^i\},\{U^{ij}\}_{{\rm C}_{ij}\in{\rm S}_i}\right\}
	\end{equation}
	where the components are shown as follows:
	
	\textit{1) Players set ${\rm S}_i$:} The aggregators ${\rm HA}_i$ and ${\rm EA}_i$ act as leaders, and offer a price respectively to the DESs. Then, ${\rm DES}_{ij}\in{\rm S}_i$ act as followers, and decide on the amount of electricity and heat they want to sell respectively according to the offered prices.
	
	\textit{2) Strategy spaces $\mathbb{P}$ and $\mathbb{D}$:} Let $\mathbb{P}=[c_e,r_e]\times[c_h,r_h]$ be the strategy space of two aggregators, where we say $\{p_e^i,p_h^i\}\in\mathbb{P}$ is a feasible strategy of ${\rm HA}_i$ and ${\rm EA}_i$. Then, let $\mathbb{D}=\times_{{\rm C}_{ij}\in{\rm S}_i}\{[0,1]\times[0,1]\}$ be the strategy space of all DESs in this city, and we have $\{\alpha^{ij},\beta^{ij}\}_{{\rm C}_{ij}\in{\rm S}_i}\in\mathbb{D}$ is a feasible strategy of DESs.
	
	\textit{3) Utility functions $\{V_e^i,V_h^i\}$ and $\{U^{ij}\}_{{\rm C}_{ij}\in{\rm S}_i}$:} Each player in this game aims to maximize its utility or profit, which reflects the quality of strategy that this player chooses. $\{V_e^i,V_h^i\}$ is the profits of aggregators, defined in (9) and (10); and $\{U^{ij}\}_{{\rm C}_{ij}\in{\rm S}_i}$ are the utilities of DESs in ${\rm S}_i$, defined in (6).

\subsection{DESs (Followers) Side Analysis}
	Given a price strategy $\{p_e^i,p_h^i\}\in\mathbb{P}$ offered by two aggregators in city ${\rm S}_i$, each ${\rm DES}_{ij}\in{\rm S}_i$ decides the amount of electricity $E_{use}^{ij}$ (resp. heat $Q_{use}^{ij}$) that sold to the ${\rm EA}_i$ (resp. ${\rm HA}_i$) by adjusting its dispatching factor $\alpha^{ij}$ (resp. $\beta^{ij}$). Thus, each ${\rm DES}_{ij}\in{\rm S}_i$ aims to choose its optimal dispatching factors $\{\alpha^{ij},\beta^{ij}\}$ according to $\{p_e^i,p_h^i\}$ by solving the following optimization problem (${\rm OP_{DES}}$), that is
	\begin{flalign}
		&\max\nolimits_{\{\alpha^{ij},\beta^{ij}\}}U^{ij}(\alpha^{ij},\beta^{ij})\\
		&\text{ s.t. } E_{use}^{ij}+Q_{use}^{ij}\geq M_{min}^{ij}\text{; }\{\alpha^{ij},\beta^{ij}\}\in[0,1]\times[0,1]
	\end{flalign}
	where $M_{min}^{ij}$ is the minimum amount of energy that is required to maintain the basic life for those residents living in this community. The objective function, shown in (6), is strictly concave and continuously differentiable, this ${\rm OP_{DES}}$ is a convex optimization problem, which will be proved later. Thus, the stationary solution is unique and optimal.
	
	To ensure reasonableness, the $E_{min}^{ij}$ should be in a valid range, thus we have $M_{min}^{ij}\in(\max\{X^{ij},Y^{ij}\},X^{ij}+Y^{ij})$. It means that this minimum requirement is larger than the production capacity of electricity or heat separately, which implies that $\alpha^{ij}$ and $\beta^{ij}$ are impossible to approach zero definitely. Hence, the restrictions on (16) can be converted equivalently to constraint (17), that is
	\begin{equation}
		X^{ij}\cdot\alpha^{ij}+Y^{ij}\cdot\beta^{ij}\geq M_{min}^{ij}\text{; }\alpha^{ij}\leq 1\text{, }\beta^{ij}\leq 1
	\end{equation}
	Then, its first-order derivatives is
	\begin{flalign}
		&\frac{\partial U^{ij}}{\partial\alpha^{ij}}=X^{ij}\cdot\bigg(\frac{k_e^{ij}b_e^{ij}}{1+b_e^{ij}X^{ij}\alpha^{ij}}-p_e^i\bigg)\\
		&\frac{\partial U^{ij}}{\partial\beta^{ij}}=Y^{ij}\cdot\bigg(\frac{k_h^{ij}b_h^{ij}}{1+b_h^{ij}Y^{ij}\beta^{ij}}-p_h^i\bigg)
	\end{flalign}
	Let $\partial U^{ij}/\alpha^{ij}=0$ and $\partial U^{ij}/\beta^{ij}=0$, we have
	\begin{equation}
		\alpha^{ij}_\circ=\frac{1}{X^{ij}}\bigg(\frac{k_e^{ij}}{p_e^i}-\frac{1}{b_e^{ij}}\bigg)\text{; }\beta^{ij}_\circ=\frac{1}{Y^{ij}}\bigg(\frac{k_h^{ij}}{p_h^i}-\frac{1}{b_h^{ij}}\bigg)
	\end{equation}
	Here, we need to note that the setting of parameter $k_e^{ij}$ (resp. $k_h^{ij}$) must be in a valid range such that $\alpha^{ij}_\circ\in(0,1)$ (resp. $\beta^{ij}_\circ\in(0,1)$) for any offered price $p_e^i\in[c_e,r_e]$ (resp. $p_h^i\in[c_h,r_h]$). Or else, this utility function is monotone, and it is meaningless to adjust its dispatching factors. Base on that, thereby we can restrict $k_e^{ij}$ and $k_h^{ij}$ as follows:
	\begin{equation}
		k_e^{ij}\in\left(\frac{r_eX^{ij}}{e-1},\frac{c_eX^{ij}}{1-1/e}\right)\text{; }k_h^{ij}\in\left(\frac{r_hY^{ij}}{e-1},\frac{c_hY^{ij}}{1-1/e}\right)
	\end{equation}
	where it assume $r_e<e\cdot c_e$ (resp. $r_h<e\cdot c_h$), or else no such $k_e^{ij}$ (resp. $k_h^{ij}$) can keep $\alpha^{ij}_\circ\in(0,1)$ (resp. $\beta^{ij}_\circ\in(0,1)$) satisfied for any offered prices.
	
	Sequentially, we use Lagrange's multipliers $\lambda_1$, $\lambda_2$ and $\lambda_3$ for constraint (17), thereby the ${\rm OP_{DES}}$, shown as (15) and (17), can be converted to the following form $L^{ij}(\alpha^{ij},\beta^{ij},\lambda_1,\lambda_2,\lambda_3)$, that is
	\begin{flalign}
	L^{ij}&=U^{ij}(\alpha^{ij},\beta^{ij})\nonumber\\
	&+\lambda_1\left(X^{ij}\cdot\alpha^{ij}+Y^{ij}\cdot\beta^{ij}-M^{ij}_{min}\right)\nonumber\\
	&+\lambda_2(1-\alpha^{ij})+\lambda_3(1-\beta^{ij})
	\end{flalign}
	where we denote $Z^{ij}=X^{ij}+Y^{ij}-M_{max}^{ij}$. Then, Based on (22), the complementary slackness conditions (KKT conditions) of ${\rm OP_{DES}}$ are demonstrated as follows:
	\begin{flalign}
		&\frac{\partial L^{ij}}{\partial\alpha^{ij}}=\frac{\partial U^{ij}}{\partial\alpha^{ij}}+\lambda_1X^{ij}-\lambda_2=0\\
		&\frac{\partial L^{ij}}{\partial\beta^{ij}}=\frac{\partial U^{ij}}{\partial\beta^{ij}}+\lambda_1Y^{ij}-\lambda_3=0\\
		&\lambda_1\left(X^{ij}\cdot\alpha^{ij}+Y^{ij}\cdot\beta^{ij}-M^{ij}_{min}\right)=0\\
		&\lambda_2(1-\alpha^{ij})=0\\
		&\lambda_3(1-\beta^{ij})=0\\
		&\lambda_1\geq 0\text{, }\lambda_2\geq 0\text{, }\lambda_3\geq 0\text{, and constraints (15)}
	\end{flalign}
	The optimal solutions of ${\rm OP_{DES}}$, shown as (15) and (17), can take one of the following four cases, that is
	
	\textit{1) Case 1:} For $\alpha^{ij}<1$ and $\beta^{ij}<1$, we have $\lambda_2=\lambda_3=0$. Look at (25), if $\lambda_1=0$, substitute it into (23) and (24), we can get a solution $\{\alpha^{ij}_\circ,\beta^{ij}_\circ\}$ according to (20). Then, we need to check whether constraint (17) can be satisfied. If yes, the optimal solution is $\{\alpha^{ij}_\circ,\beta^{ij}_\circ\}$. If no, it means $\lambda_1>0$ and $X^{ij}\cdot\alpha^{ij}+Y^{ij}\cdot\beta^{ij}-M_{min}^{ij}=0$. At this time, by solving (23) and (24), we have
	\begin{flalign}
		&\alpha^{ij}_\Diamond=\frac{1}{X^{ij}}\bigg(\frac{k_e^{ij}}{p_e^i-\lambda_1}-\frac{1}{b_e^{ij}}\bigg)\\
		&\beta^{ij}_\Diamond=\frac{1}{Y^{ij}}\bigg(\frac{k_h^{ij}}{p_h^i-\lambda_1}-\frac{1}{b_h^{ij}}\bigg)
	\end{flalign}
	Substitute (29) and (30) into (25),
	\begin{equation}
		A^{ij}{\lambda_1}^2+B^{ij}\lambda_1+C^{ij}=0
	\end{equation}
	where $A^{ij}=M_{min}^{ij}+1/b_e^{ij}+1/b_h^{ij}$, $B^{ij}=k_e^{ij}+k_h^{ij}-A^{ij}(p_e^i+p_h^i)$, and $C^{ij}=A^{ij}p_e^ip_h^i-k_e^{ij}p_h^i-k_h^{ij}p_e^i$. By solving (31), we have two solutions, they are
	\begin{equation}
		\lambda_1=\frac{-B^{ij}\pm\sqrt{{B^{ij}}^2-4A^{ij}C^{ij}}}{2A^{ij}}
	\end{equation}
	Here, it is easy to verify $B^{ij}<0$ and $C^{ij}>0$ based on (9), (21), and $M_{min}^{ij}\in(\max\{X^{ij},Y^{ij}\},X^{ij}+Y^{ij})$, thus it is possible that $\Delta^{ij}={B^{ij}}^2-4A^{ij}C^{ij}<0$. If $\Delta^{ij}<0$, there is no real solution; else we need to check whether the $\lambda_1$ defined on (32) satisfies $\lambda_1>0$. If $\lambda_1>0$, substitute (32) into (29) and (30), we obtain a solution $\{\alpha^{ij}_\Diamond,\beta^{ij}_\Diamond\}$. If it is feasible, namely $\alpha^{ij}_\Diamond,\beta^{ij}_\Diamond<1$, the optimal solution can be determined by $\{\alpha^{ij}_\Diamond,\beta^{ij}_\Diamond\}$.
	
	\textit{2) Case 2:} For $\alpha^{ij}=1$ and $\beta^{ij}<1$, we have $\lambda_3=0$. Look at (25), if $\lambda_1=0$, substitute it into (24), we can get a solution $\{1,\beta^{ij}_\circ\}$ according to (20). Then, we need to check whether constraint (17) can be satisfied and $\lambda_2={\partial U^{ij}}/{\partial\alpha^{ij}}\geq0$. If yes, the optimal solution $\{1,\beta^{ij}_\circ\}$. If no, it means $\lambda_1>0$ and $Y^{ij}\cdot\beta^{ij}+X^{ij}-M_{min}^{ij}=0$. According to (24), we have $\beta^{ij}_\square$ which is shown as (30). Substitute (30) into (25),
	\begin{equation}
		\bigg(\frac{k_h^{ij}}{p_h^i-\lambda_1}-\frac{1}{b_h^{ij}}\bigg)+X^{ij}-M_{min}^{ij}=0
	\end{equation}
	By solving (33), we have
	\begin{equation}
		\lambda_1=p_h^i-\frac{k_h^{ij}b_h^{ij}}{b_h^{ij}(M_{min}^{ij}-X^{ij})+1}
	\end{equation}
	If $\lambda_1>0$ and $\lambda_2={\partial U^{ij}}/{\partial\alpha^{ij}}+\lambda_1X^{ij}\geq0$, substitute (34) into (30), we obtain a solution $\{1,\beta^{ij}_\square\}$. If we have $\beta^{ij}_\square<1$, the optimal solution can be determined by $\{1,\beta^{ij}_\square\}$.
	
	\textit{3) Case 3:} For $\alpha^{ij}<1$ and $\beta^{ij}=1$, we have $\lambda_2=0$. Look at (25), if $\lambda_1=0$, substitute it into (23), we can get a solution $\{\alpha^{ij}_\circ,1\}$ according to (18). Then, we need to check whether constraint (17) can be satisfied and $\lambda_3={\partial U^{ij}}/{\partial\beta^{ij}}\geq0$. If yes, the optimal solution is $\{\alpha^{ij}_\circ,1\}$. If no, it means $\lambda_1>0$ and $X^{ij}\cdot\alpha^{ij}+Y^{ij}-M_{min}^{ij}=0$. According to (23), we have $\alpha^{ij}_\square$ which is shown as (29). Substitute (29) into (25),
	\begin{equation}
		\bigg(\frac{k_e^{ij}}{p_e^i-\lambda_1}-\frac{1}{b_e^{ij}}\bigg)+Y^{ij}-M_{min}^{ij}=0
	\end{equation}
	By solving (35), we have
	\begin{equation}
		\lambda_1=p_e^i-\frac{k_e^{ij}b_e^{ij}}{b_e^{ij}(M_{min}^{ij}-Y^{ij})+1}
	\end{equation}
	If $\lambda_1>0$ and $\lambda_3={\partial U^{ij}}/{\partial\beta^{ij}}+\lambda_1Y^{ij}\geq0$, substitute (36) into (29), we obtain a solution $\{\alpha^{ij}_\square,1\}$. If we have $\alpha^{ij}_\square<1$, the optimal solution can be determined by $\{\alpha^{ij}_\square,1\}$.
	
	\textit{4) Case 4:} For $\alpha^{ij}=1$ and $\beta^{ij}=1$, we have $\lambda_1=0$ because we have assumed $M_{min}^{ij}\leq X^{ij}+Y^{ij}$ before. Substitute it into (23) and (24), we have
	\begin{flalign}
		&\left(\alpha^{ij}_\Diamond=1\right)=\frac{1}{X^{ij}}\bigg(\frac{k_e^{ij}X^{ij}}{\lambda_2+p_e^iX^{ij}}-\frac{1}{b_e^{ij}}\bigg)\\
		&\left(\beta^{ij}_\Diamond=1\right)=\frac{1}{Y^{ij}}\bigg(\frac{k_h^{ij}Y^{ij}}{\lambda_3+p_h^iY^{ij}}-\frac{1}{b_h^{ij}}\bigg)
	\end{flalign}
	By solving (37) and (38), we have
	\begin{equation}
	\lambda_2=\frac{k_e^{ij}X^{ij}}{X^{ij}+1/b_e^{ij}}-p_e^iX^{ij}\text{; } \lambda_3=\frac{k_h^{ij}Y^{ij}}{Y^{ij}+1/b_h^{ij}}-p_h^iY^{ij}
	\end{equation}
	According to (9) (21), the maximum value of $\lambda_2$ can be obtained when giving $k_e^{ij}=c_eX^{ij}/(1-1/e)$. Substitute it into (39), we have $\lambda_2<c_e-p_e^i\leq0$ because of $p_e^i\in[c_e,r_e]$. By using the same way, we have $\lambda_3<c_h-p_h^i\leq0$ because of $p_h^i\in[c_h,r_h]$ as well. It does not satisfy (28), thus this solution $\{1,1\}$ is not feasible and cannot occur.
	
	To sum up, offered a price strategy $\{p_e^i,p_h^i\}\in\mathbb{P}$ by aggregators, the optimal response of each ${\rm C}_{ij}\in{\rm S}_i$ will be obtained by above procedure. It is one of the three cases, except case 4, that depends on the offered prices, minimum requirement $M_{min}^{ij}$, and choice of satisfaction cofficient $k_e^{ij}$ and $k_h^{ij}$. Since the expressions of the solution is very complicated, we cannot give a unified formal expression to summarize the results that contains all cases.
	
\subsection{Aggregators (Leaders) Side Analysis}
	After receiving the responses $E_{exc}^{ij}$ (resp. $Q_{exc}^{ij}$) of all ${\rm C}_{ij}$ in city ${\rm S}_i$, the profit gained by aggregators ${\rm EA}_i$ (resp. ${\rm HA}_i$) can be determined according to (10) and (11). They assume each ${\rm DES}_{ij}\in{\rm S}_i$ will respond to them with the optimal strategy according to their offered price. Thus, ${\rm EA}_i$ and ${\rm HA}_i$ aim to choose its optimal prices $\{p_e^i,p_h^i\}$ by solving the following optimization problem (${\rm OP_{AGS}}$), that is
	\begin{flalign}
		&\max\nolimits_{\{p_e^i\}}V_e^i(p_e^i,p_h^i)\text{ s.t. }p_e^i\in[c_e,r_e]\\
		&\max\nolimits_{\{p_h^i\}}V_h^i(p_h^i,p_e^i)\text{ s.t. }p_h^i\in[c_h,r_h]
	\end{flalign}
	where ${\rm EA}_i$ (resp. ${\rm HA}_i$) attempts to select an optimal price $p^{i*}_e$ (resp. $p^{i*}_h$) to maximize its profit given $p_h^i$ (resp. $p_e^i$). The objective function, shown in (10) (resp. (11)), is strictly concave and continuous differentiable with respect to $p^{i}_e$ (resp. $p^{i}_h$), which will be proved later.
	
	First, we consider electricity aggregator ${\rm EA}_i$ alone. Feed a price $p_e^i$ into $V_e^i(\cdot,p_h^i)$, the response $\alpha^{ij}(p^i_e,p_h^i)$ of each ${\rm DES}_{ij}\in{\rm S}_i$ must be in one of the following four events: (1) $\alpha^{ij}=\alpha^{ij}_\circ$; (2) $\alpha^{ij}=1$; (3) $\alpha^{ij}=\alpha^{ij}_\Diamond$; and (4) $\alpha^{ij}=\alpha^{ij}_\square$. Then, its first order derivatives is
	\begin{flalign}
		{\partial\alpha^{ij}}/{\partial p_e^i}&=-k_e^{ij}/\left(X^{ij}(p_e^i)^2\right)\\
		&=0\\
		&=-k_e^{ij}/\left(X^{ij}(p_e^i-\lambda_1)^2\right)\text{, }\lambda_1=\text{(31)}\\
		&=-k_e^{ij}/\left(X^{ij}(p_e^i-\lambda_1)^2\right)\text{, }\lambda_1=\text{(34)}
	\end{flalign}
	where it is one-to one correspondences between (42)$-$(45) and event (1)$-$(4). Then, the first-order derivative of ${\rm EA}_i$'s objective function is
	\begin{equation}
		\frac{\partial V_e^{i}}{\partial p_e^i}=-\sum_{{\rm C}_{ij}\in{\rm S}_i}X^{ij}\left((1-\alpha^{ij})+(r_e-p_e^i)\frac{\partial\alpha^{ij}}{\partial p_e^i}\right)
	\end{equation}
	Combined with (42)$-$(45), let ${\partial V_e^{i}}/{\partial p_e^i}=0$, we can get a solution $\hat{p}_{e}^i$ that maximizes $V_e^i(\cdot,p_h^i)$ given $p_h^i$. However, this $\hat{p}_{e}^i$ is constrained on the range of $[c_e,r_e]$, thus the optimal price strategy $\bar{p}_{e}^i$ of ${\rm EA}_i$ is shown as follows:
	\begin{equation}
		\bar{p}_{e}^i=\left
		\{\begin{IEEEeqnarraybox}[\relax][c]{l's}
		r_e,&if $\hat{p}_{e}^i\geq r_e$\\
		c_e,&if $\hat{p}_{e}^i\leq c_e$\\
		\hat{p}_{e}^i,&if $c_e<\hat{p}_{e}^i<r_e$%
		\end{IEEEeqnarraybox}
		\right.
	\end{equation}
	Due to the fact that the profit function $V_e^i(\cdot,p_h^i)$ given $p_h^i$ is strictly concave with respect to $p_e^i$, it increases first and then decreases with the increase of $p_e^i$. Thus, the maximum profit is obtained at the price of $r_e$ when $\hat{p}_{e}^i\geq r_e$; Similarly, the maximum profit is obtained at the price of $c_e$ when $\hat{p}_{e}^i\leq c_e$; else obtained at stationary point.
	
	Then, we consider heat aggregator ${\rm HA}_i$ alone. Feed a price $p_h^i$ into $V^i_h(\cdot,p_e^i)$, the respone $\beta^{ij}(p^i_e,p_h^i)$ of each ${\rm DES}_{ij}\in{\rm S}_i$ must be in one of the following four events: (1) $\beta^{ij}=\beta^{ij}_\circ$; (2) $\beta^{ij}=1$; (3) $\beta^{ij}=\beta^{ij}_\Diamond$; and (4) $\beta^{ij}=\beta^{ij}_\square$. Then, its first order derivative ${\partial\beta^{ij}}/{\partial p_h^i}$ can be computed by replacing $X^{ij}$ with $Y^{ij}$, $k_e^{ij}$ with $k_h^{ij}$, and $p_e^i$ with $p_h^i$ in (42)-(45), corresponding to event (1)$-$(4). Then, the first-order derivative of ${\rm HA}_i$'s objective function is
	\begin{equation}
		\frac{\partial V_h^{i}}{\partial p_h^i}=-\sum_{{\rm C}_{ij}\in{\rm S}_i}Y^{ij}\left((1-\beta^{ij})+(r_h-p_h^i)\frac{\partial\beta^{ij}}{\partial p_h^i}\right)
	\end{equation}
	Let ${\partial V_h^{i}}/{\partial p_h^i}=0$, we can get a solution $\hat{p}_{h}^i$ that maximizes $V_h^i(\cdot,p_e^i)$ given $p_e^i$. Similar to (47), constrained on the range of $[c_h,r_h]$, the optimal price $\bar{p}_{h}^i$ of ${\rm HA}_i$ can be formulated similar to the analysis of (47) from its concavity.
	
\subsection{Stackalberg Equilibrium}
	The aggregators, ${\rm EA}_i$ and ${\rm HA}_i$ in a smart city ${\rm S}_i$, play a non-cooperative game with each other to offer the unit prices for electricity and heat. They all want to maximize their profit according to their profit function defined on (10) and (11). We denote this game between aggregators by $\mathbb{A}=\{\{{\rm HA}_i,{\rm EA}_i\},\mathbb{P},\{V_e^i,V_h^i\}\}$ and introduce the concept of the Nash equilibrium (NE) shown as follows:
	\begin{defn}[Nash Equilibrium]
		Given a game $\mathbb{A}$ defined as above, a feasible price strategy $\{\tilde{p}_{e}^i,\tilde{p}_{h}^i\}\in\mathbb{P}$ is the Nash equilibrium if no player can improve its profit by changing its strategy unilaterally, that is
		\begin{equation}
			V_e^i(\tilde{p}_{e}^i,\tilde{p}_{h}^i)\geq V_e^i(p_e^i,\tilde{p}_{h}^i)\text{; }V_h^i(\tilde{p}_{h}^i,\tilde{p}_{e}^i)\geq V_h^i(p_h^i,\tilde{p}_{e}^i)
		\end{equation}
	\end{defn}

	There is a property that at the NE, no aggregator attempts to offer a new price again because they all achieve their mutually satisfactions respectively. Then, we need to study the existence and uniqueness of the NE of game $\mathbb{A}$ between two aggregators in a city.
	\begin{lem}
		The Nash equilibirum of game $\mathbb{A}$ between aggregators always exist and is unique.
	\end{lem}
	\begin{proof}
		The strategy space in game $\mathbb{A}$ has been denoted by $\mathbb{P}=[c_e,r_e]\times[c_h,r_h]$, which is a convex, closed, and non-empty subset of the space $\mathbb{R}^2$. Take aggregate ${\rm EA}_i$ as an example, $\alpha^{ij}$ is the responsive dispatching factor given the offered price $\{p_e^i,p_h^i\}\in\mathbb{P}$ from community ${\rm C}_{ij}\in{\rm S}_i$. From (42)$-$(45), its second-order derivatives is
		\begin{flalign}
			{\partial^2\alpha^{ij}}/{\partial {p_e^i}^2}&=2k_e^{ij}/\left(X^{ij}(p_e^i)^3\right)\\
			&=0\\
			&=2k_e^{ij}/\left(X^{ij}(p_e^i-\lambda_1)^3\right)\text{, }\lambda_1=\text{(29)}\\
			&=2k_e^{ij}/\left(X^{ij}(p_e^i-\lambda_1)^3\right)\text{, }\lambda_1=\text{(32)}
		\end{flalign}
		where they correspond to event (1)$-$(4). Then, the second-order derivative of ${\rm EA}_i$'s objective function is
		\begin{equation}
			\frac{\partial^2 V_e^{i}}{\partial {p_e^i}^2}=\sum_{{\rm C}_{ij}\in{\rm S}_i}X^{ij}\left(2\cdot\frac{\partial\alpha^{ij}}{\partial p_e^i}-(r_e-p_e^i)\frac{\partial^2\alpha^{ij}}{\partial {p_e^i}^2}\right)
		\end{equation}
		Here, observe that ${\partial\alpha^{ij}}/{\partial p_e^i}\leq 0$ from (42)$-$(45), and ${\partial^2 V_e^{i}}/{\partial {p_e^i}^2}\geq 0$ from (50)$-$(53), we have ${\partial^2 V_e^{i}}/{\partial {p_e^i}^2}\leq 0$. Thus, $V_e^{i}(\cdot,p_h^i)$ is concave with respect to $p_e^i$.
		
		Consider aggregator ${\rm HA}_i$ and $\beta^{ij}$ from community ${\rm C}_{ij}\in{\rm S}_i$, its second-order derivative ${\partial^2\beta^{ij}}/{\partial {p_e^i}^2}$ can be computed by replacing $X^{ij}$ with $Y^{ij}$, $k_e^{ij}$ with $k_h^{ij}$, and $p_e^i$ with $p_h^i$ in (48)$-$(51). Then, the second-order derivative of ${\rm HA}_i$'s objective function is
		\begin{equation}
			\frac{\partial^2 V_h^{i}}{\partial {p_h^i}^2}=\sum_{{\rm C}_{ij}\in{\rm S}_i}Y^{ij}\left(2\cdot\frac{\partial\beta^{ij}}{\partial p_h^i}-(r_h-p_h^i)\frac{\partial^2\beta^{ij}}{\partial {p_h^i}^2}\right)
		\end{equation}
		By similar analysis, we have ${\partial^2 V_h^{i}}/{\partial {p_h^i}^2}\leq 0$ and $V_h^{i}(\cdot,p_e^i)$ is concave with respect to $p_h^i$. Thus, game $\mathbb{A}$ is a concave 2-person game. Because of their concavity, the Nash equilibrium exists and is unique according to \cite{rosen1965existence}.
	\end{proof}
	In a smart city ${\rm S}_i$, the aggregators offer the price strategy $\{p_e^i,p_h^i\}\in\mathbb{P}$ in the first stage, then each ${\rm DES}_{ij}\in{\rm S}_i$ decides its optimal dispatching strategy $\{\alpha^{ij},\beta^{ij}\}$ according to the offered prices in the second stage. It formulates a MLMF Stackelberg game $\mathbb{G}$ between aggregators and DESs, shown as (14). The optimal strategy set $\{\{\tilde{p}_e^i,\tilde{p}_h^i\},\{\gamma^{il}\}_{{\rm C}_{il}\in{\rm S}_i}\}$, where $\{\gamma^{il}\}$ is the optimal response of community ${\rm C}_{il}\in{\rm S}_i$ based on its previous leaders' prices, can be obtained at the Stackelberg equilibrium (SE), defined as follows:
	\begin{defn}[Stackelberg Equilibrium]
		Given a game $\mathbb{G}$ defined as (14), a feasible strategy $\{\{\tilde{p}_e^i,\tilde{p}_h^i\},\{\gamma^{il}\}_{{\rm C}_{il}\in{\rm S}_i}\}$ is the Stackelberg equilibrium if no player, including leaders and followers, can improve its utility or profit by changing its strategy unilaterally, that is
		\begin{flalign}
			&\small{U^{ij}\left(\tilde{p}^i,\{\gamma^{il}\}_{{\rm C}_{il}\in{\rm S}_i}\right)
			\geq U^{ij}\left(\tilde{p}^i,\bar{\gamma}^{ij}\cup\{\gamma^{il}\}_{{\rm C}_{il}\in{\rm S}_i\backslash{\rm C}_{ij}}\right)}\\
			&V_e^i\left(\tilde{p}^i,\{\gamma^{il}\}_{{\rm C}_{il}\in{\rm S}_i}\right)\geq V_e^i\left(\{p_e^i,\tilde{p}_h^i\},\{\gamma^{il}\}_{{\rm C}_{il}\in{\rm S}_i}\right)\\
			&V_h^i\left(\tilde{p}^i,\{\gamma^{il}\}_{{\rm C}_{il}\in{\rm S}_i}\right)\geq V_h^i\left(\{\tilde{p}_e^i,p_h^i\},\{\gamma^{il}\}_{{\rm C}_{il}\in{\rm S}_i}\right)
		\end{flalign}
		where we denote prices $\tilde{p}^i=\{\tilde{p}_e^i,\tilde{p}_h^i\}$ and $\bar{\gamma}^{ij}$ is any feasible strategy of ${\rm DES}_{ij}$.
	\end{defn}

	After reaching the SE, none of them tends to change its strategy again because they cannot improve their utilities or profits further by changing unilaterally. Then, we need to study the existence and uniqueness of the SE of game $\mathbb{G}$ between aggregators and DESs in a city.
	
	\begin{algorithm}[!t]
		\caption{\text{Find NE}}\label{a1}
		\begin{algorithmic}[1]
			\renewcommand{\algorithmicrequire}{\textbf{Input:}}
			\renewcommand{\algorithmicensure}{\textbf{Output:}}
			\REQUIRE Game $\mathbb{G}$ in a city ${\rm S}_i$ and a small step $\Delta$
			\ENSURE Price strategy $\{\tilde{p}_e^i,\tilde{p}_h^i\}$
			\STATE Initialize: $\{\tilde{p}_e^i,\tilde{p}_h^i\}\leftarrow\{c_e,c_h\}$
			\WHILE {True}
			\STATE $\{x,y\}\leftarrow\{\tilde{p}_e^i,\tilde{p}_h^i\}$
			\STATE // Consider aggregator ${\rm EA}_i$
			\IF {$V_e^i(\tilde{p}_{e}^i+\Delta,\tilde{p}_{h}^i)\geq V_e^i(\tilde{p}_{e}^i,\tilde{p}_{h}^i)$ and $V_e^i(\tilde{p}_{e}^i+\Delta,\tilde{p}_{h}^i)\geq V_e^i(\tilde{p}_{e}^i-\Delta,\tilde{p}_{h}^i)$}
			\STATE $\tilde{p}_{e}^i\leftarrow\min\{r_e,\tilde{p}_{e}^i+\Delta\}$
			\ELSIF {$V_e^i(\tilde{p}_{e}^i-\Delta,\tilde{p}_{h}^i)\geq V_e^i(\tilde{p}_{e}^i,\tilde{p}_{h}^i)$ and $V_e^i(\tilde{p}_{e}^i-\Delta,\tilde{p}_{h}^i)\geq V_e^i(\tilde{p}_{e}^i+\Delta,\tilde{p}_{h}^i)$}
			\STATE $\tilde{p}_{e}^i\leftarrow\max\{c_e,\tilde{p}_{e}^i-\Delta\}$
			\ENDIF
			\STATE // Consider aggregator ${\rm HA}_i$
			\IF {$V_h^i(\tilde{p}_{h}^i+\Delta,\tilde{p}_{e}^i)\geq V_h^i(\tilde{p}_{h}^i,\tilde{p}_{e}^i)$ and $V_h^i(\tilde{p}_{h}^i+\Delta,\tilde{p}_{e}^i)\geq V_h^i(\tilde{p}_{h}^i-\Delta,\tilde{p}_{e}^i)$}
			\STATE $\tilde{p}_{h}^i\leftarrow\min\{r_h,\tilde{p}_{h}^i+\Delta\}$
			\ELSIF {$V_h^i(\tilde{p}_{h}^i-\Delta,\tilde{p}_{e}^i)\geq V_h^i(\tilde{p}_{h}^i,\tilde{p}_{e}^i)$ and $V_h^i(\tilde{p}_{h}^i-\Delta,\tilde{p}_{e}^i)\geq V_h^i(\tilde{p}_{h}^i+\Delta,\tilde{p}_{e}^i)$}
			\STATE $\tilde{p}_{h}^i\leftarrow\max\{c_h,\tilde{p}_{h}^i-\Delta\}$
			\ENDIF
			\IF {$\{x,y\}=\{\tilde{p}_e^i,\tilde{p}_h^i\}$}
			\STATE Break
			\ENDIF
			\STATE // Reduce $\Delta$ once for each iteration
			\STATE $\Delta\leftarrow\delta\cdot\Delta$
			\ENDWHILE
			\RETURN $\{\tilde{p}_e^i,\tilde{p}_h^i\}$
		\end{algorithmic}
	\end{algorithm}

	\begin{thm}
		The Stackelberg equilibirum of game $\mathbb{G}$ between aggregators and DESs always exist and is unique.
	\end{thm}
	\begin{proof}
		In this game $\mathbb{G}$, the aggregators will offer prices to those DESs to purchase energies in their city first. From Lemma 1, a Nash equilibrium always exists and is unique between aggregators. According to the aforementioned analysis on DESs side, they are able to respond to aggregators with their optimal dispatching strategies based on the offered prices and their restrictions. Thus, the Stackelberg equilibrium always exists and is unique.
	\end{proof}

\subsection{Distributed Algorithm}
In order to find the NE between aggregators based on the optimal responses from DESs, we adopt the sub-gradient technique \cite{boyd2004convex} \cite{xiao2011simple} \cite{zhang2016multi} for determining price strategies. It is shown in Algorithm \ref{a1}. In Algorithm \ref{a1}, each aggregators is assigned with its lowest price. At this time, the dispatching factors of each DES in this city is closest to one, which should be in feasible space defined on (17). Then, consider aggregator ${\rm EA}_i$, in each iteration, it updates its price in manner of increasing by $\Delta$ or decreasing by $\Delta$, where $\Delta$ is a given small step. According to current prices $\{\tilde{p}_e^i,\tilde{p}_h^i\}$, we compare the profits of ${\rm EA}_i$ by offering a price $\tilde{p}_e^i$, $\tilde{p}_e^i+\Delta$, and $\tilde{p}_e^i-\Delta$, then choose the best one and update the price $\tilde{p}_e^i$. Consider aggregator ${\rm HA}_i$ similarly, we compare the profits of ${\rm HA}_i$ by offering a price $\tilde{p}_h^i$, $\tilde{p}_h^i+\Delta$, and $\tilde{p}_h^i-\Delta$, then choose the best one and update the price $\tilde{p}_h^i$. At last, we update $\Delta$ with $\delta\cdot\Delta$ where $\delta\in(0,1)$ is a attenuation factor.

\begin{thm}
	Given an initial price strategy and step $\Delta$, the Nash equilibrium of game $\mathbb{A}$ can be obtained by the sub-gradient algorithm, shown in Algorithm \ref{a1}.
\end{thm}
\begin{proof}
	Based on the conclusion of \cite{boyd2004convex} \cite{xiao2011simple}, the sub-gradient algorithm can converge to an optimal solution in convex optimization. The objective functions of aggregators are concave, thus they cannot improve its profit by changing strategies unilaterally when reaching the optimal solution.
\end{proof}

\section{Numerical Simulations}
	In this section, we conduct several experiments to model price competition and energies trading in a smart city.
\subsection{Simulation Setup}
Consider a city ${\rm S}=\{\{{\rm EA},{\rm HA}\},\{{\rm C}_1,{\rm C}_2,\cdots,{\rm C}_n\}\}$, we denote the number of communities in this city by $n$. At the standard atmosphere, the calorific value of natural gas is $q=3.6\times10^7$ ${\rm J}/{\rm m}^3$ on average. The retail price of electricity in U.S. is $0.2$ ${\rm dollar}/{\rm kw}\cdot{\rm h}$. According to the conversion relationship of $1$ ${\rm kw}\cdot{\rm h}=3.6\times10^6$ ${\rm J}$, it is $5.5\times10^{-8}$ ${\rm dollar}/{\rm J}$. Equivalently, we regard it as $r_e=5.5\times10^{-8}$ ${\rm coin}/{\rm J}$ in our B-METS. The electric conversion of GT is $\eta_g=0.5$. We define the maximum gas consumption $F_m=200$ ${\rm m}^3/{\rm day}$ and its unit price $c_f=1.08$ ${\rm coin}/{\rm m}^3$. Thus, we have $c_e=3.00\times10^{-8}$ ${\rm coin}/{\rm J}$ and $p_e\in[3.00\times10^{-8},5.50\times10^{-8}]$ for the EA definitely. The efficiency of heat recovery system is given by $\eta_r=0.8$, thereby we have $c_h=3.75\times10^{-8}$ ${\rm coin}/{\rm J}$. The retail price of heat is $r_h=6.25\times10^{-8}$ ${\rm coin}/{\rm J}$. Thus, we have $p_h\in[3.75\times10^{-8},6.25\times10^{-8}]$ for the HA definitely. According to (9), we have $b_e=4.773\times10^{-10}$ and $b_h=5.966\times10^{-10}$. Then according to (21), we have $k_e\in[115.24,170.85]$ and $k_h\in[104.76,170.85]$.

\begin{figure}[!t]
	\centering
	\subfigure[EA's profit $V_e$ under $k_1$]{
		\includegraphics[width=0.48\linewidth]{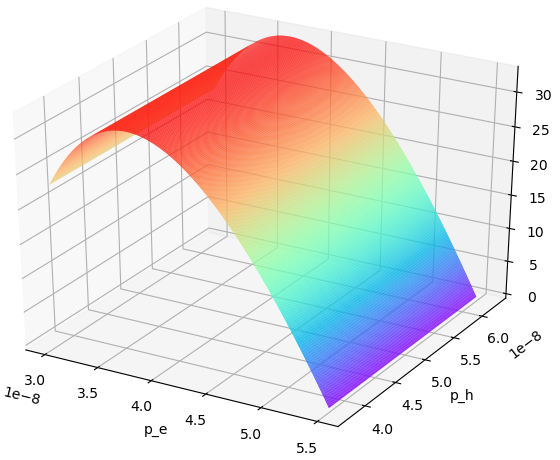}
	}%
	\subfigure[HA's profit $V_h$ under $k_1$]{
		\includegraphics[width=0.48\linewidth]{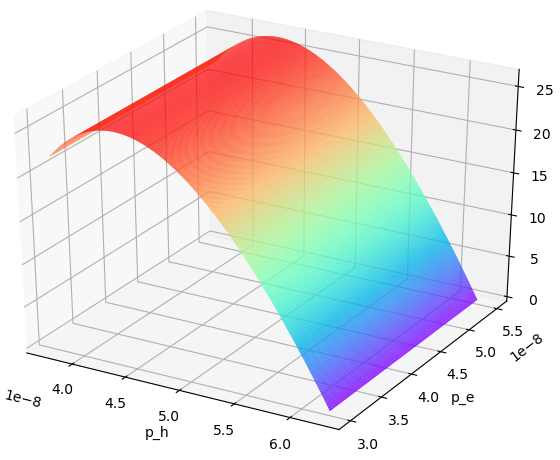}
	}%

	\subfigure[DES's utility $U$ under $k_1$]{
		\centering
		\includegraphics[width=0.48\linewidth]{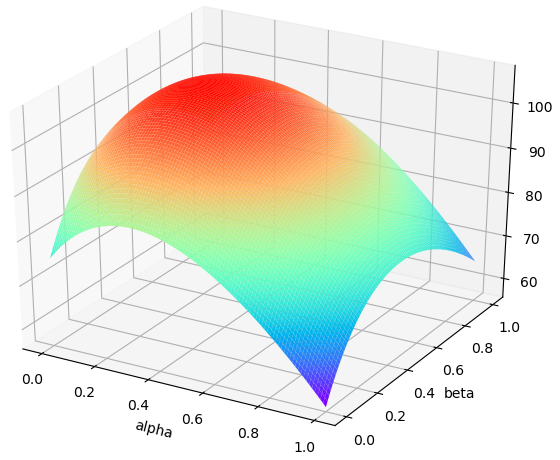}
	}%
	\subfigure[EA's profit $V_e$ under $k_2$]{
		\centering
		\includegraphics[width=0.48\linewidth]{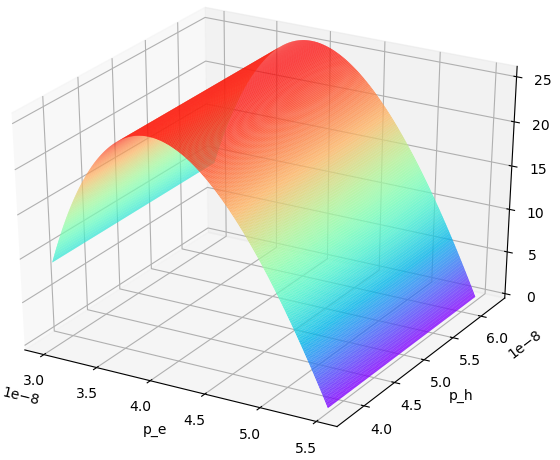}
	}%

	\subfigure[HA's profit $V_h$ under $k_2$]{
		\centering
		\includegraphics[width=0.48\linewidth]{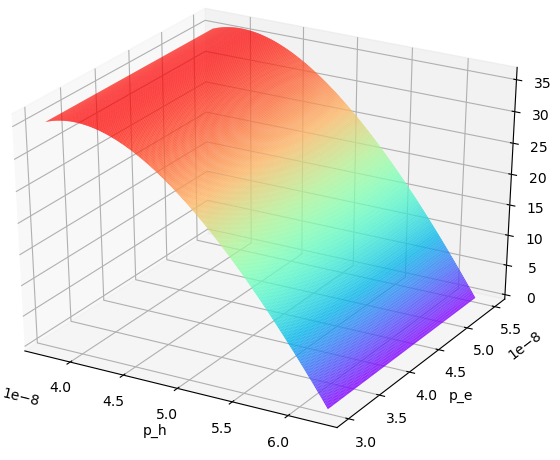}
	}%
	\subfigure[DES's utility $U$ under $k_2$]{
		\centering
		\includegraphics[width=0.48\linewidth]{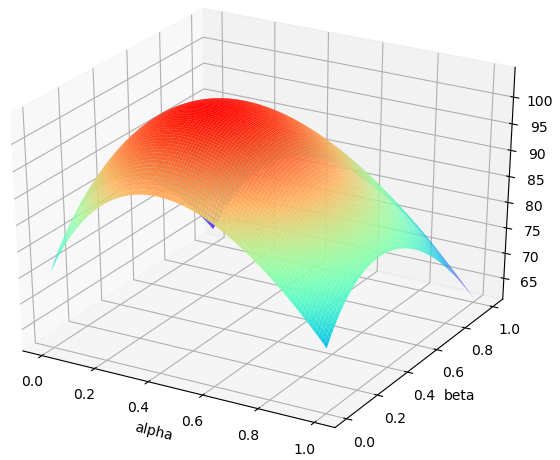}
	}%
	\centering
	\caption{The objective function of entities, including EA, HA, and DES, in city ${\rm S}$ under the setting $k_1$ and $k_2$.}
	\label{fig4}
\end{figure}

\subsection{Simulation Results}
\textit{1) Concavity of functions: } Consider a city ${\rm S}$ that has only one community, we define this DES's satisfaction coefficient under two settings $k_1=(k_e,k_h)=(143.05,137.81)$ and $k_2=(k_e,k_h)=(159.73,117.98)$. Fig. \ref{fig4} draws the objective function of entities in city $S$ under the two settings, where we define the minimum energy restriction at ${\rm OP_{DES}}$ as $M_{min}=0$. It means that there is no restriction on DES to choose their partition coefficients in order to demonstrate complete functional properties. Let us look at (a) (b) (c) in Fig. \ref{fig4}. Shown as (a), as $p_e$ increases, EA's profit function increases first and then decreases under any price $p_h$ offered by HA, and its objective value has nothing to do with $p_h$. There is no competitiveness between EA and HA because of no restriction. It proves the profit function $V_e(\cdot,p_h)$ is concave with respect to $p_e$. Similarly, shown as (b), we have HA's profit function $V_h(\cdot,p_e)$ is concave with respect to $p_h$ as well. From (c), it is the utility function of this DES according to offered prices $p_e=p_h=4.5\times10^{-8}$, which prove the concavity of DES's utility function. It shows that our previous analysis about DESs' response is valid, and DESs are always able to respond aggregators with the optimal strategy according to the prices offered by aggregators. If a Nash equilibrium exists between aggregators, then a Stackelberg equilibrium among players in game $\mathbb{G}$ must exist definitely.

\textit{2) Effect of satisfaction coefficients: } Shown as (d) (e) (f) in Fig. \ref{fig4}, under the setting $k_2$, shown as (d) (e) (f) in Fig. \ref{fig4}, we increase $k_e$ but decrease $k_b$. Shown as (d), as $k_e$ increases, the maximum point that obtains the maximum profit for EA moves toward the positive direction. It implies that the EA has to offer a higher price to buy electricity from DES in order to gain the maximum profit, because electricity used to serve community can contribute more utility than before. Similarly, shown as (e), as $k_h$ decreases, the maximum point that obtains the maximum profit for HA moves toward the negative direction. From (f), as $k_e$ increases and $k_h$ decreases, the maximum point that obtains the maximum utility for DES according to $p_e=p_h=4.5\times10^{-8}$ moves from $(\alpha_\circ,\beta_\circ)=(0.301,0.481)$ in (c) to $(0.404,0.328)$ in (f). Thus, we have $\alpha_\circ$ (resp. $\alpha_\circ$) increases with the growth of $k_e$ (resp. $k_h$).

\begin{figure}[!t]
	\centering
	\subfigure[EA's profit function $V_e(\cdot,p_h)$]{
		\includegraphics[width=0.48\linewidth]{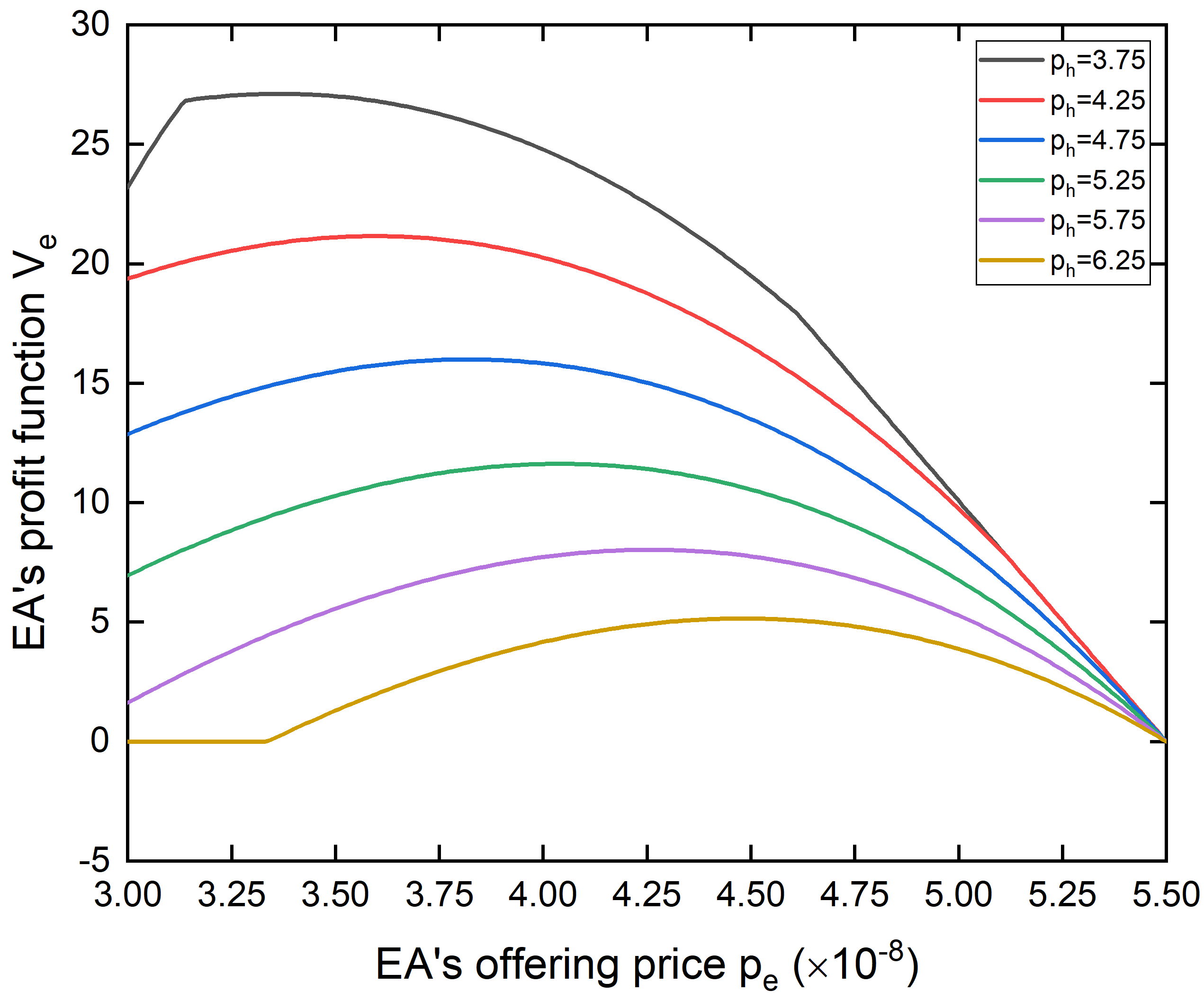}
	}%
	\subfigure[DES's optimal response $\alpha$]{
		\includegraphics[width=0.48\linewidth]{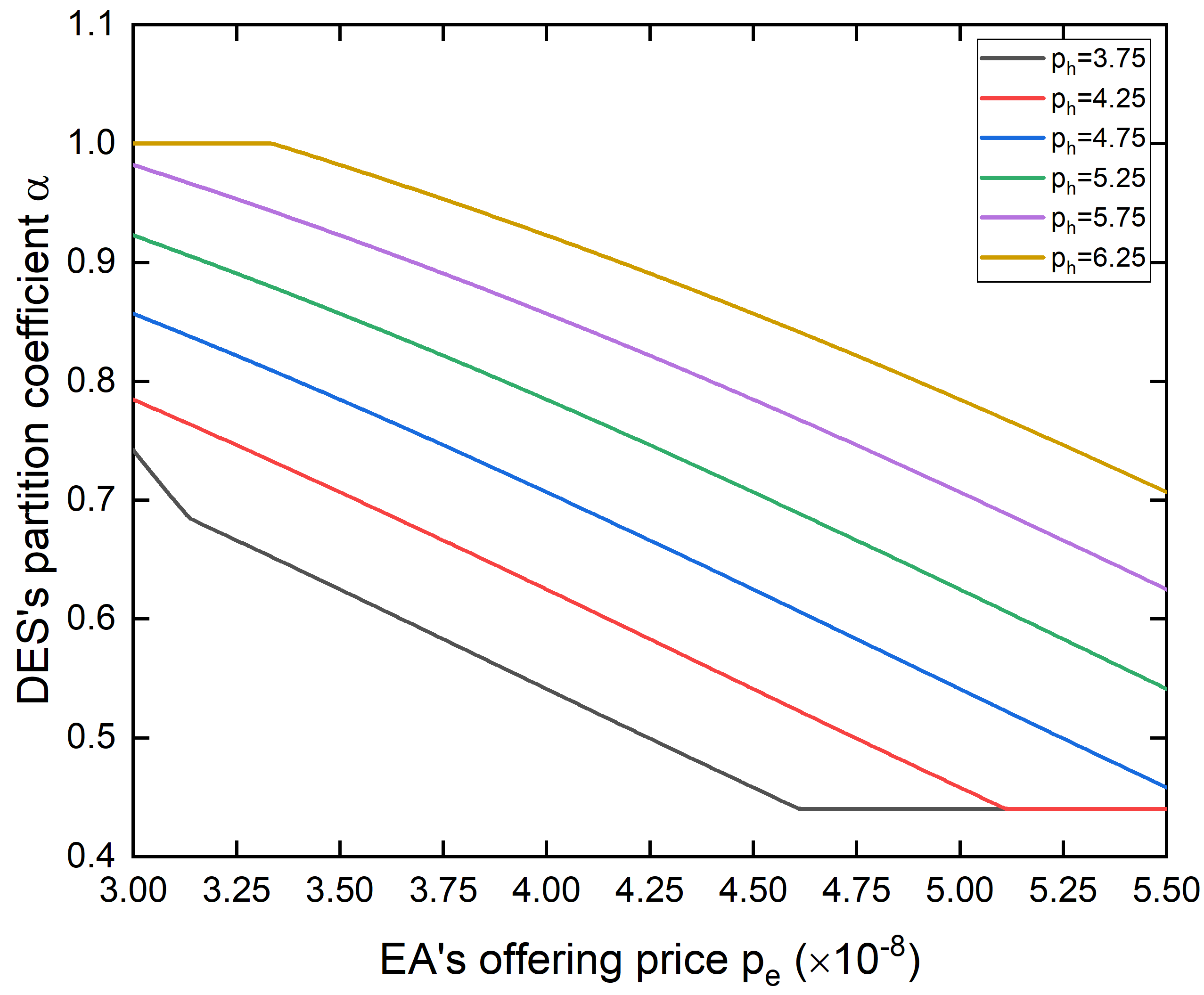}
	}%

	\subfigure[DES's optimal response $\beta$]{
		\centering
		\includegraphics[width=0.48\linewidth]{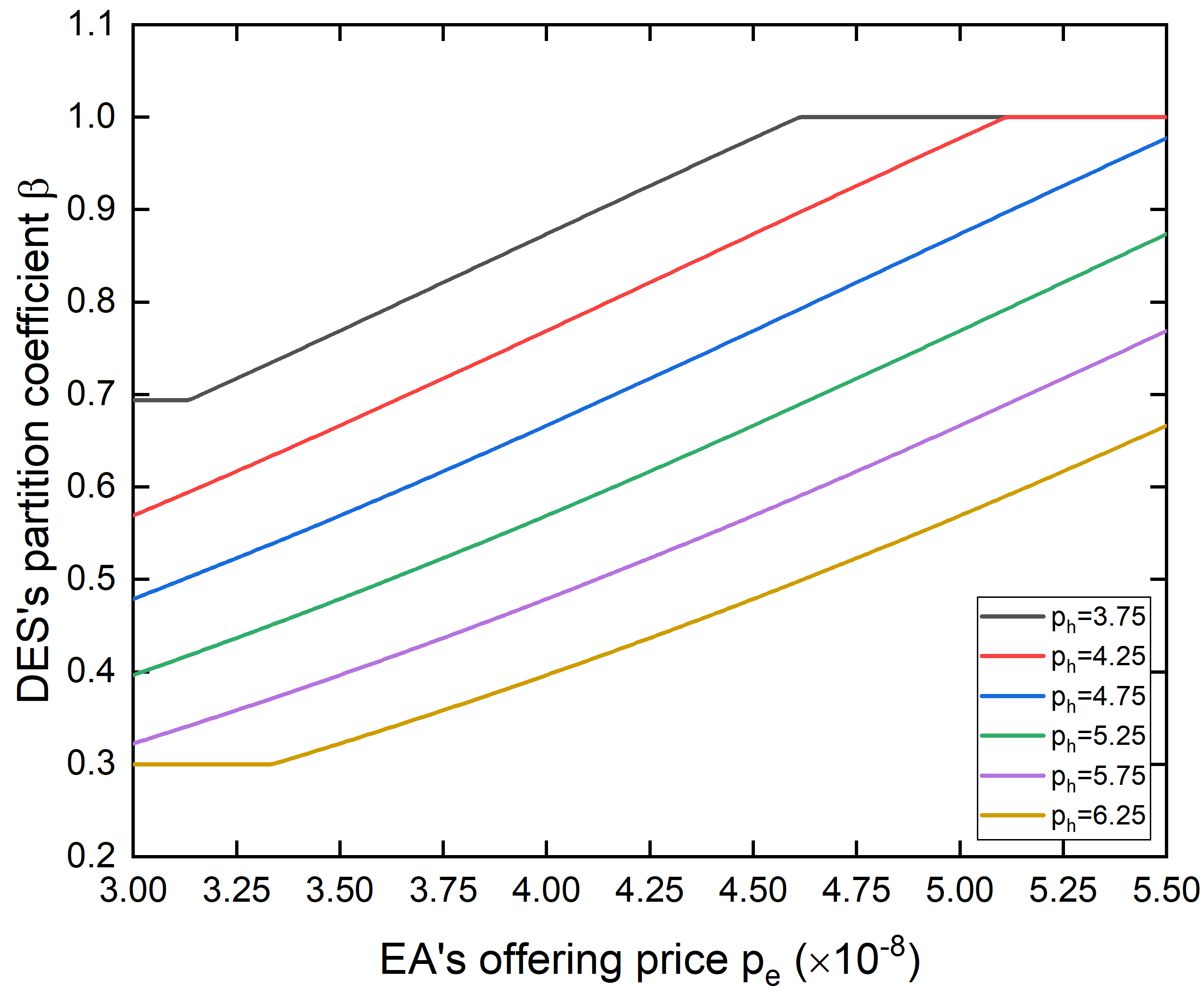}
	}%
	\subfigure[HA's profit function $V_h(\cdot,p_e)$]{
		\includegraphics[width=0.48\linewidth]{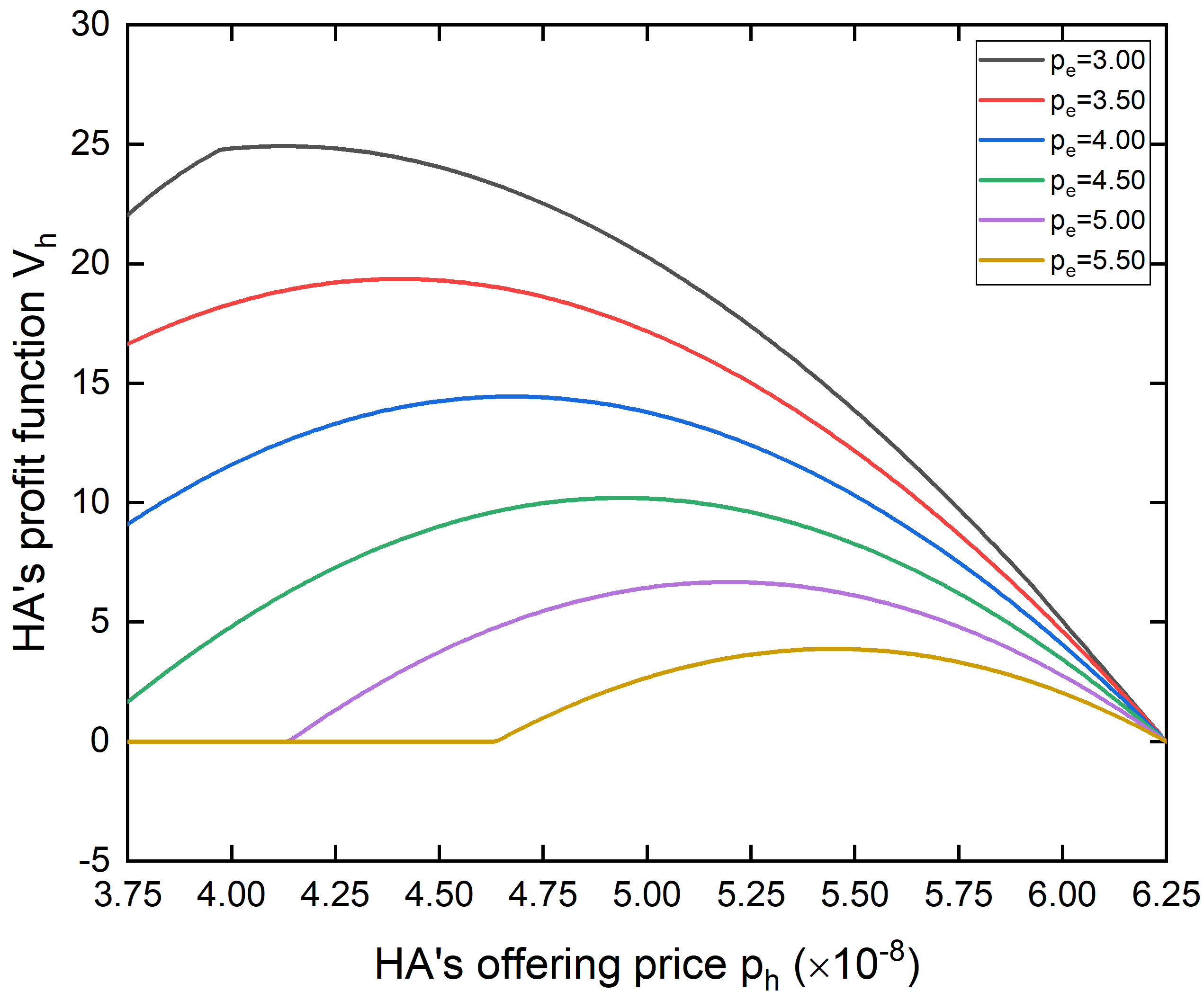}
	}%

	\subfigure[DES's optimal response $\alpha$]{
		\centering
		\includegraphics[width=0.48\linewidth]{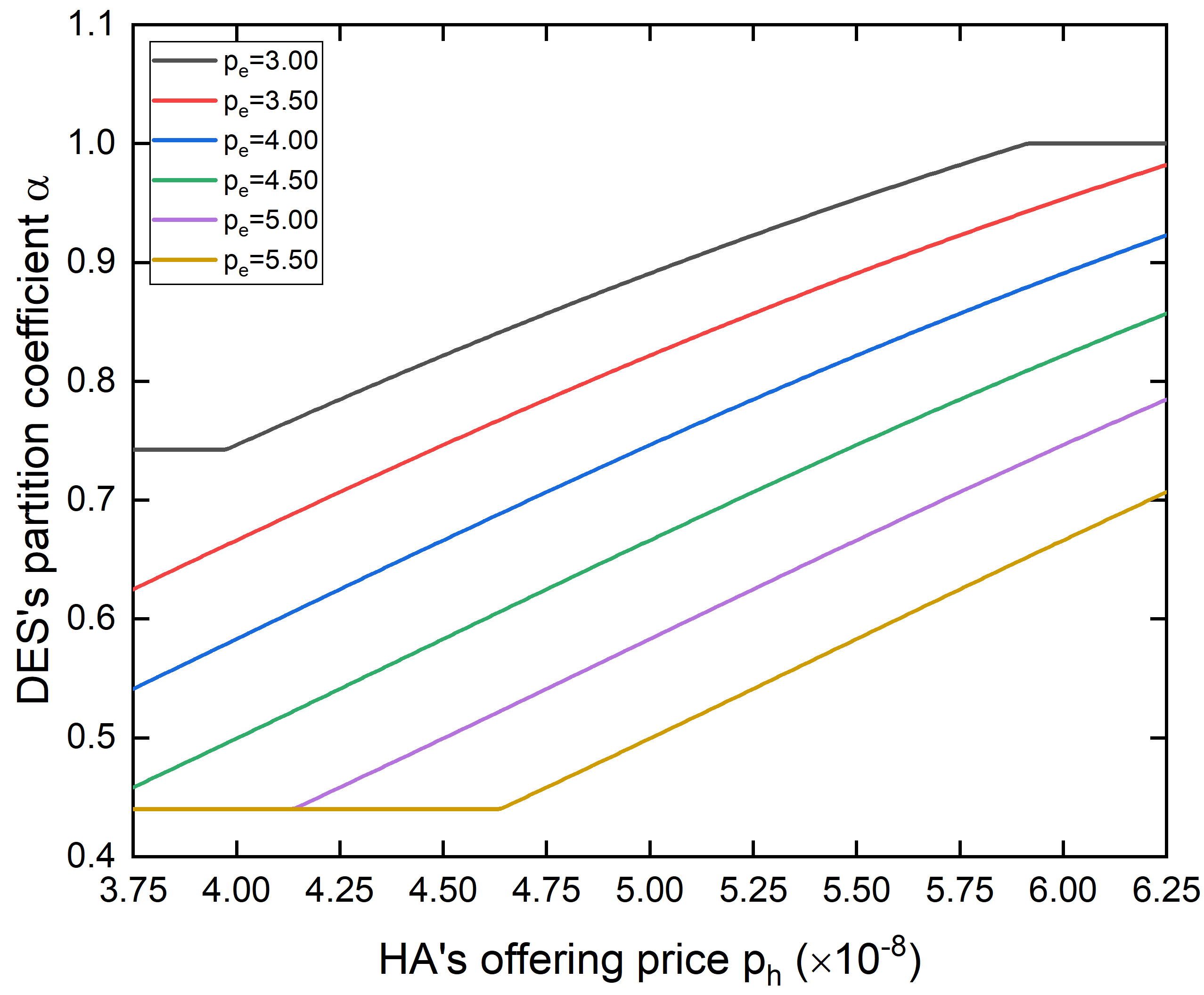}
	}%
	\subfigure[DES's optimal response $\beta$]{
		\includegraphics[width=0.48\linewidth]{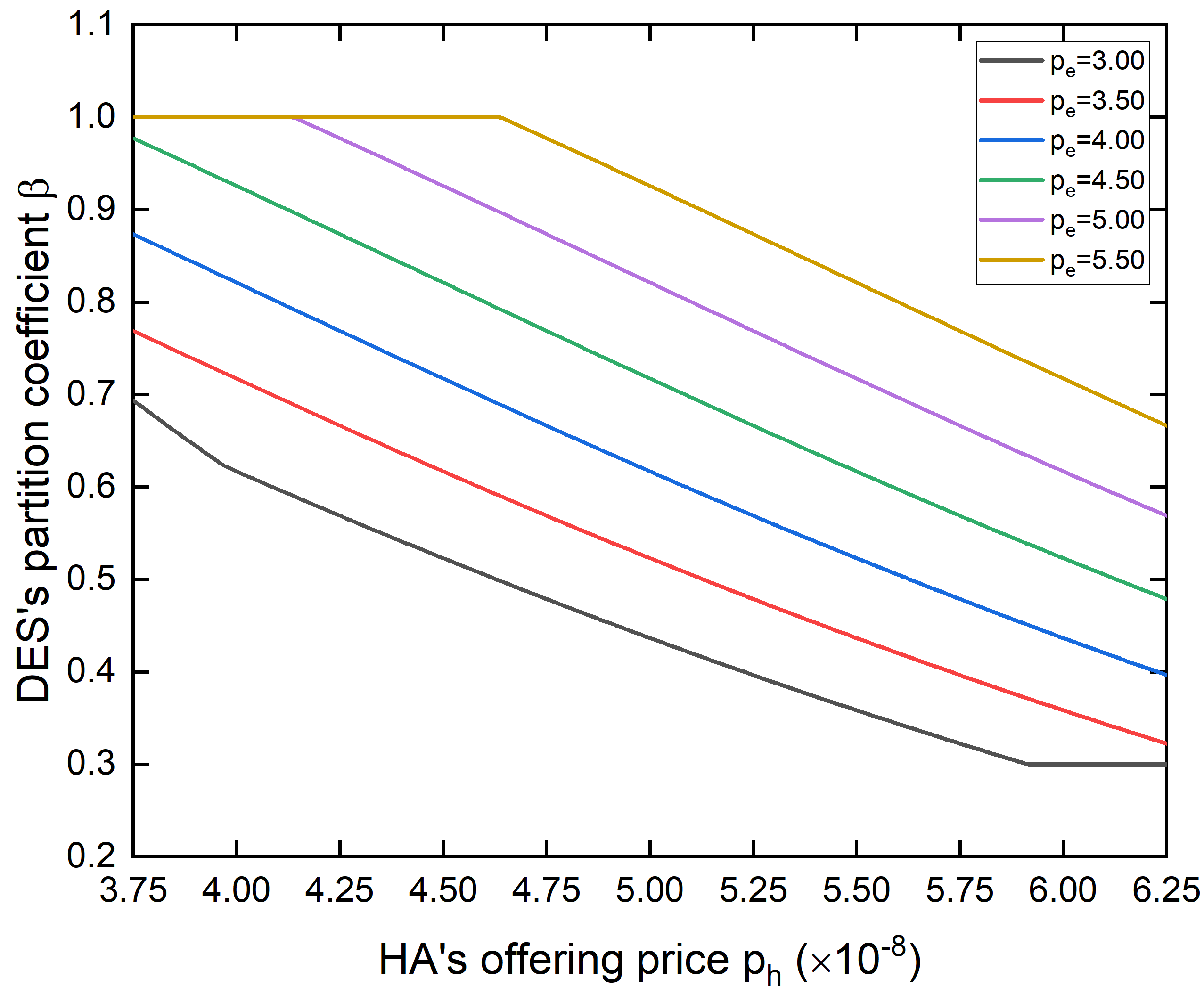}
	}%
	\centering
	\caption{The objective functions of aggregators and DES's optimal responses in city ${\rm S}$ under the setting $M_1$.}
	\label{fig5}
\end{figure}

\begin{figure}[!t]
	\centering
	\subfigure[EA's profit function $V_e(\cdot,p_h)$]{
		\includegraphics[width=0.48\linewidth]{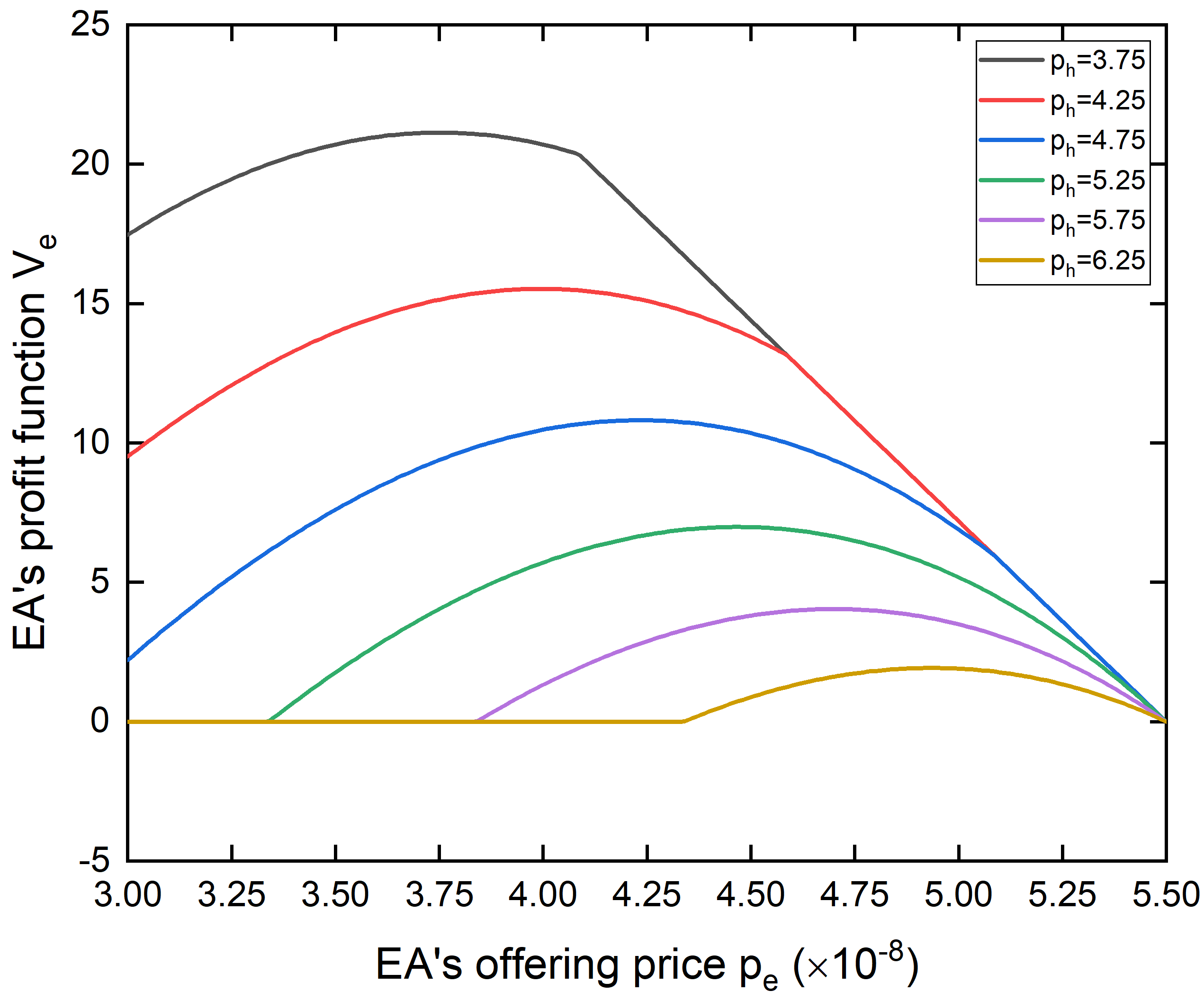}
	}%
	\subfigure[DES's optimal response $\alpha$]{
		\includegraphics[width=0.48\linewidth]{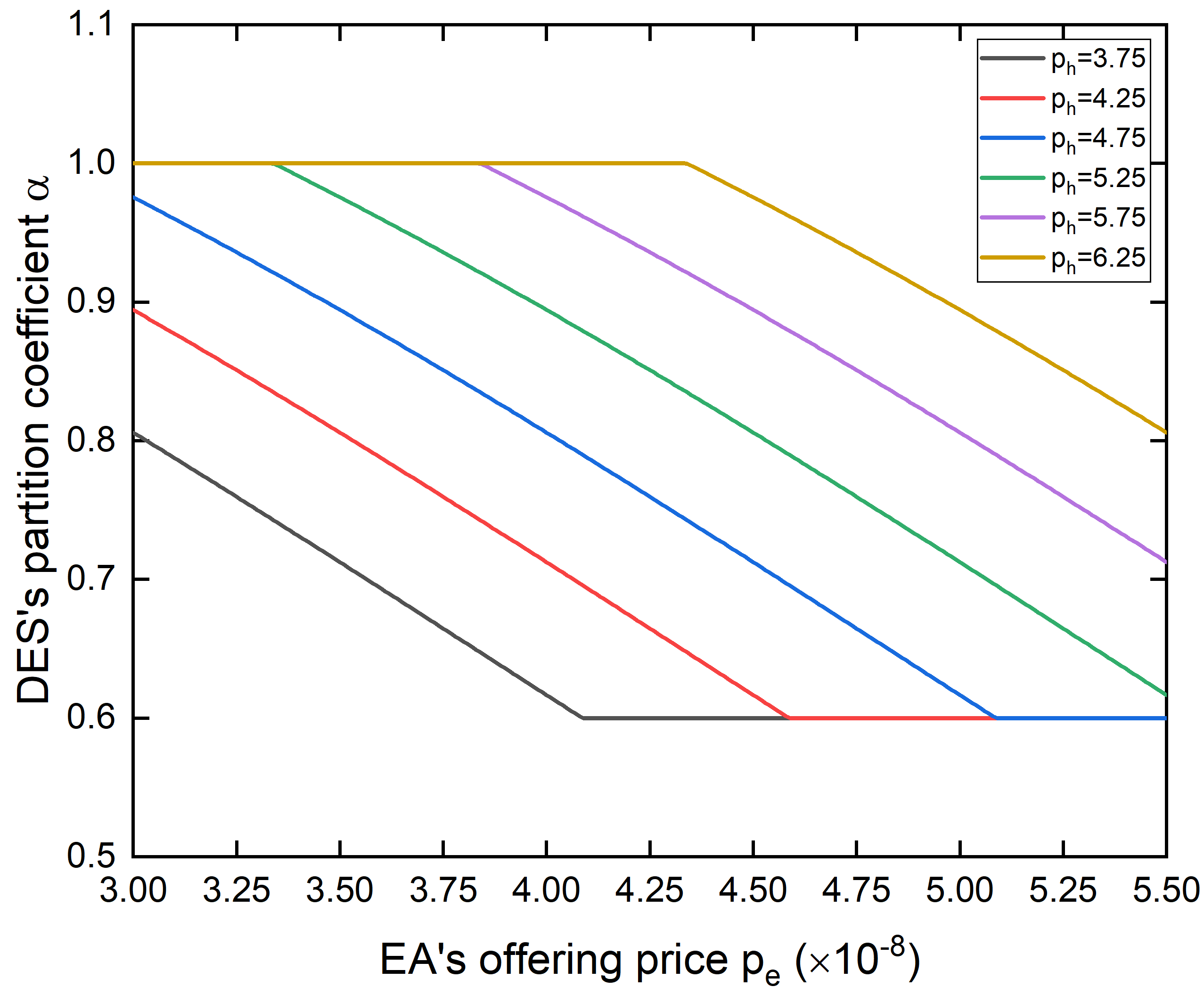}
	}%
	
	\subfigure[DES's optimal response $\beta$]{
		\centering
		\includegraphics[width=0.48\linewidth]{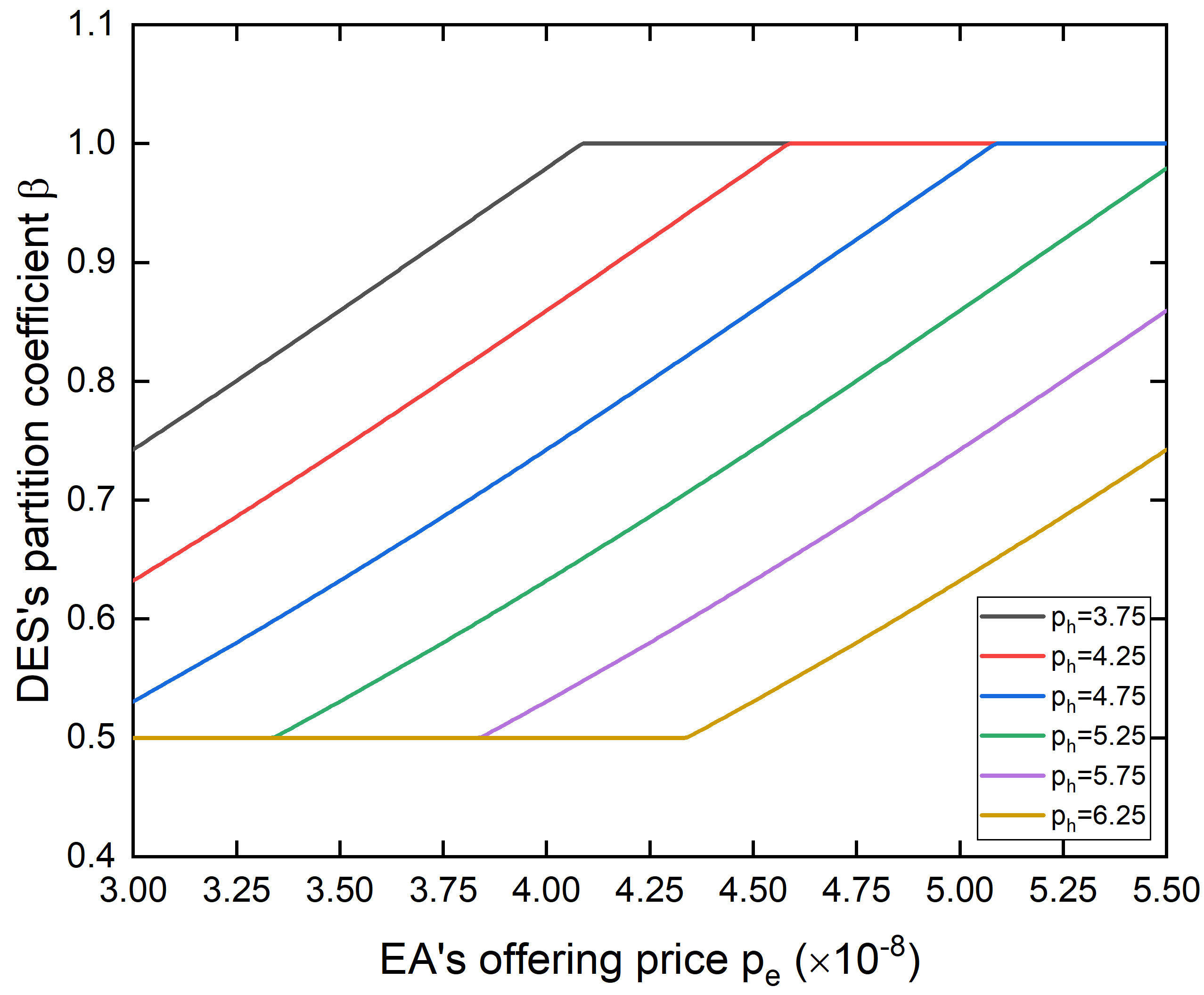}
	}%
	\subfigure[HA's profit function $V_h(\cdot,p_e)$]{
		\includegraphics[width=0.48\linewidth]{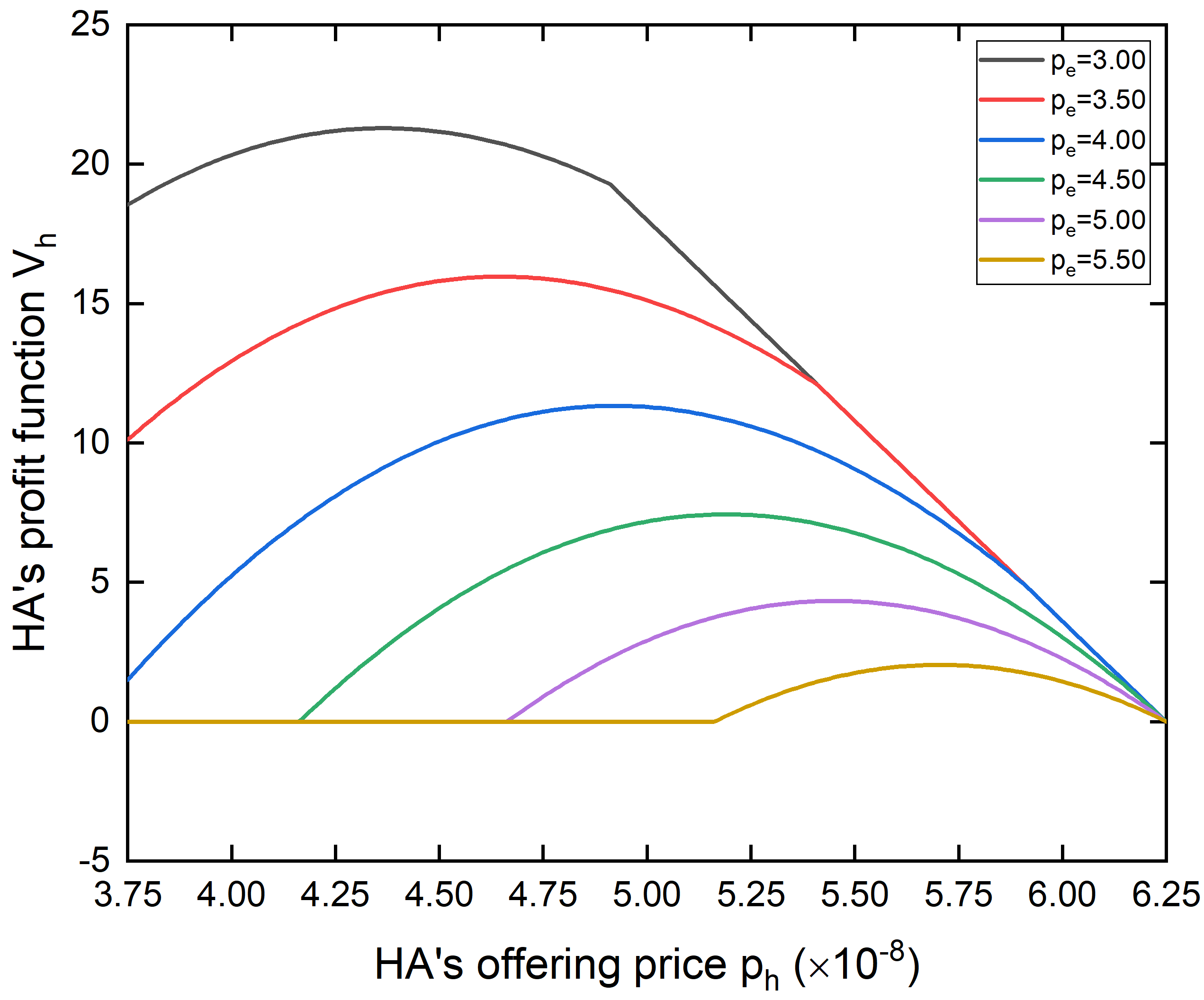}
	}%
	
	\subfigure[DES's optimal response $\alpha$]{
		\centering
		\includegraphics[width=0.48\linewidth]{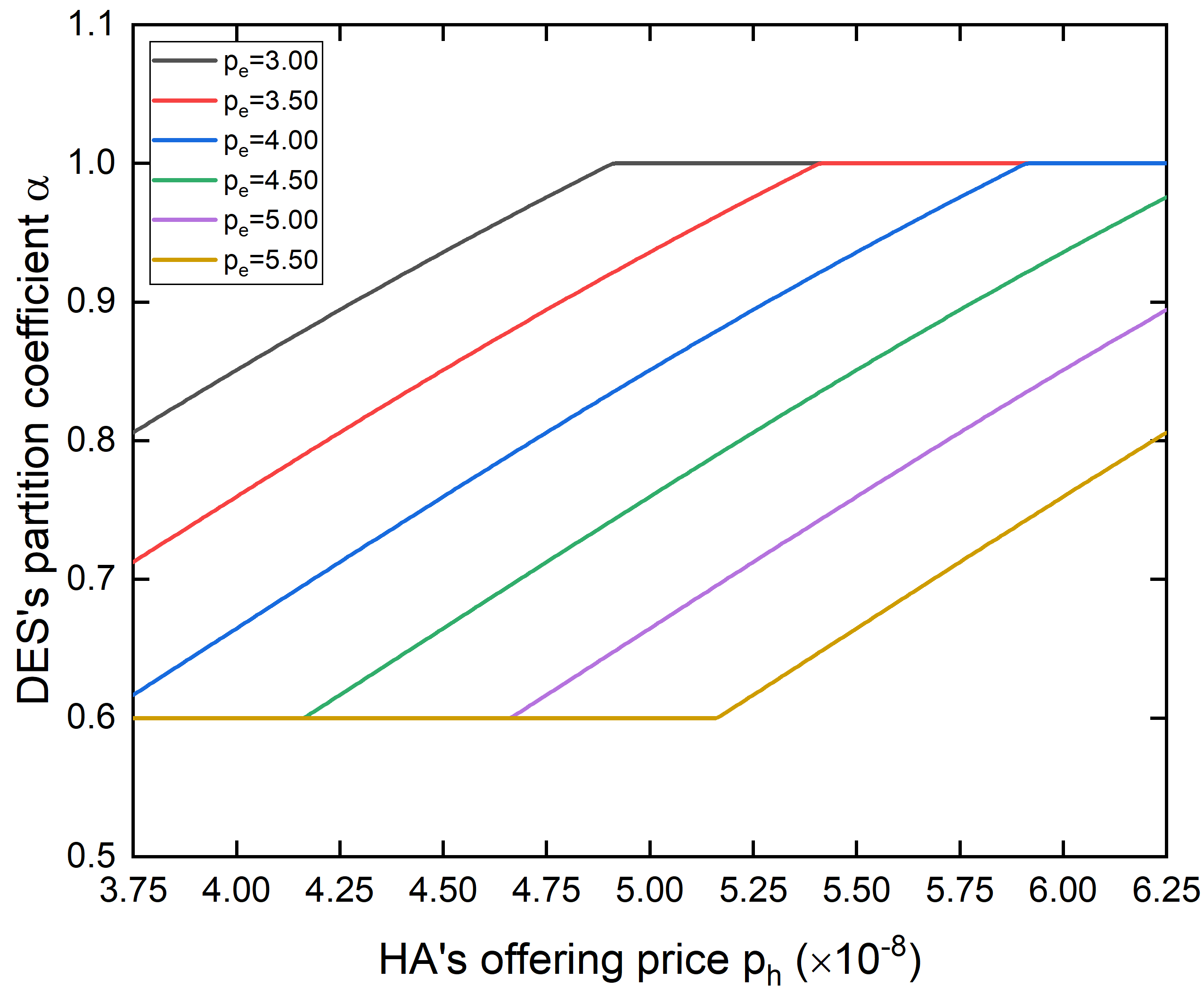}
	}%
	\subfigure[DES's optimal response $\beta$]{
		\includegraphics[width=0.48\linewidth]{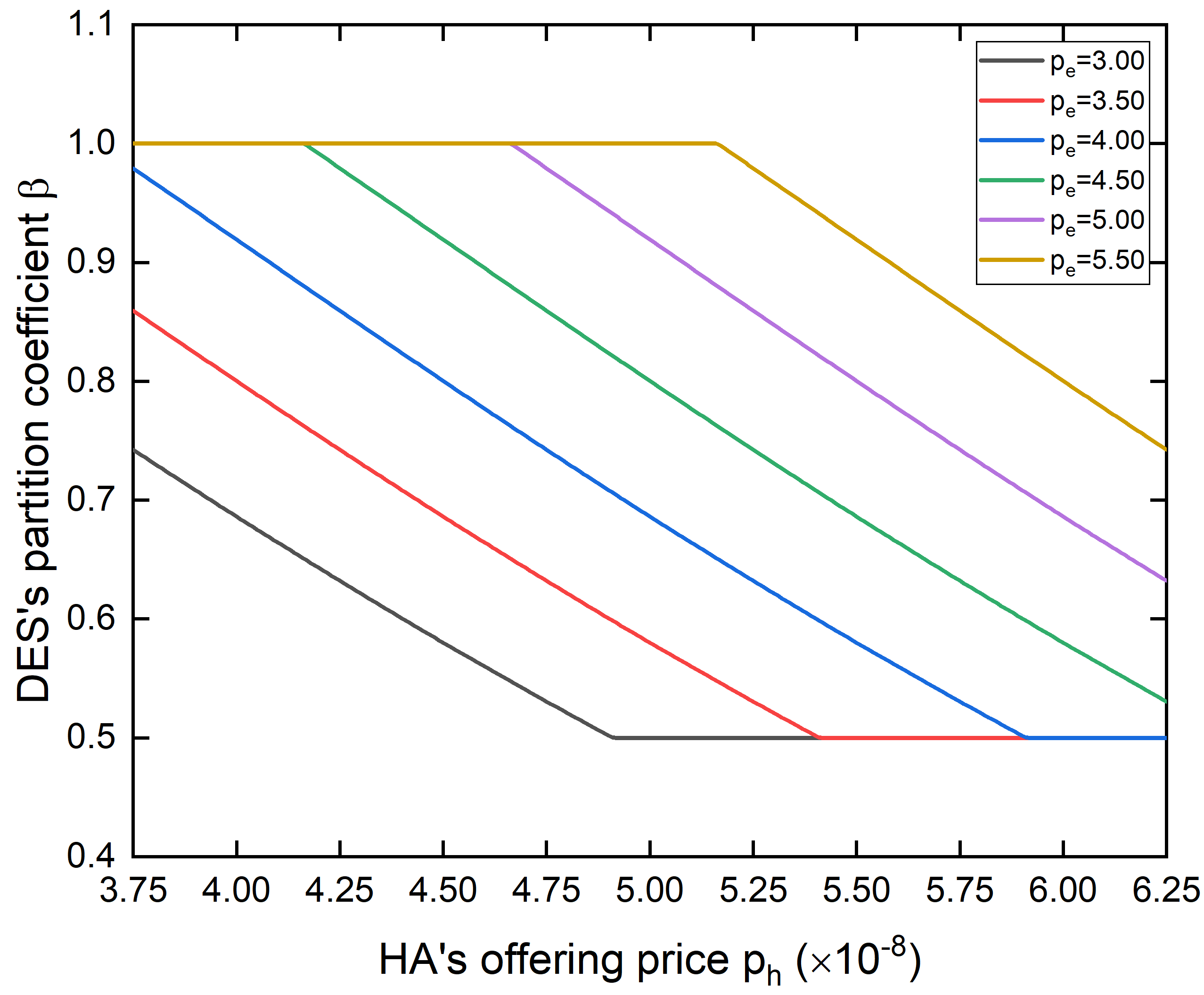}
	}%
	\centering
	\caption{The objective functions of aggregators and DES's optimal responses in city ${\rm S}$ under the setting $M_2$.}
	\label{fig6}
\end{figure}

\textit{3) Effect of restrictions: } From the definition of ${\rm OP_{DES}}$, we have a restriction that requires a feasible solution must satisfy $X\cdot\alpha+Y\cdot\beta\geq M_{min}$ and $M_{min}\in(\max\{X,Y\},X+Y)$. Here, we all adopt satisfaction coefficient $k_1$. In this part, we consider two different settings, that is $M_1=0.7\times\max\{X,Y\}+0.3\times(X+Y)$ and $M_2=0.5\times\max\{X,Y\}+0.5\times(X+Y)$ where $M_1<M_2$. In order to demonstrate the effect of restrictions clearly, we use 2D figures instead of 3D figures. Fig. \ref{fig5} and Fig. \ref{fig6} draw the objective functions of aggregators and optimal responses of DES according to aggregators' offered prices in city ${\rm S}$ under the two settings. Let us discuss the typical black curve, shown as (a) with $p_h=3.75\times10^{-8}$ in Fig. \ref{fig5}. At the beginning, EA's profit increases from $p_e=3\times10^{-8}$ to $3.13\times10^{-8}$ where DES's $\alpha$ decreases but $\beta$ keeps constant. This response implies that DES's optimal strategy can be obtained at $(\alpha_\circ,\beta_\circ)$ where the first-order derivative is equal to zero. Then, the middle section is a smooth curve, where DES's response is at the tight border $X\cdot\alpha+Y\cdot\beta=M_1$. In this section, we can see DES's response $\alpha$ decreases linearly and $\beta$ increases linearly. Finally starting from $p_e=4.5\times10^{-8}$, EA's profit decreases linearly because of $\beta=1$. Let us look at the yellow curve, shown as (a) with $p_h=6.25\times10^{-8}$ in Fig. \ref{fig5}. At the beginning, EA's profit is equal to zero since DES responds with $(\alpha=1,\beta=0.3)$. Due to the high price offered by EA and low price offered by EA, all electricity should be partitioned to meet the minimum energy restriction and only sell heat for making revenue. For Fig. \ref{fig6}, compared with Fig. \ref{fig6}, we find that these functions show some structural changes as $M_{min}$ increases. The restriction $M_2$ in Fig. \ref{fig6} is larger than $M_1$ in Fig. \ref{fig5}, which indicates the DES has to use more energy to serve its community. Shown as (a) (d) in Fig. \ref{fig6}, we can know that the DES's optimal strategy cannot be obtained at stationary points $(\alpha_\circ,\beta_\circ)$. All DES's responses are at the tight border $X\cdot\alpha+Y\cdot\beta=M_2$. Besides, the sections of $\alpha=1$ or $\beta=1$ are much larger than that under the restriction $M_1$. However, no matter what $M_{min}$ is, the EA (HA) always needs to offer an increasing price in order to get its maximum profit as the price offered by HA (EA) increases. This reflects the competition between aggregators, which is different from that there is no restriction in Fig. \ref{fig4}. Therefore, the restriction settings in ${\rm OP_{DES}}$ have significant effects on the objective functions of aggregators and optimal responses of DES.

\begin{figure}[!t]
	\centering
	\subfigure[Initialization: $(r_e,r_h)$]{
		\includegraphics[width=0.48\linewidth]{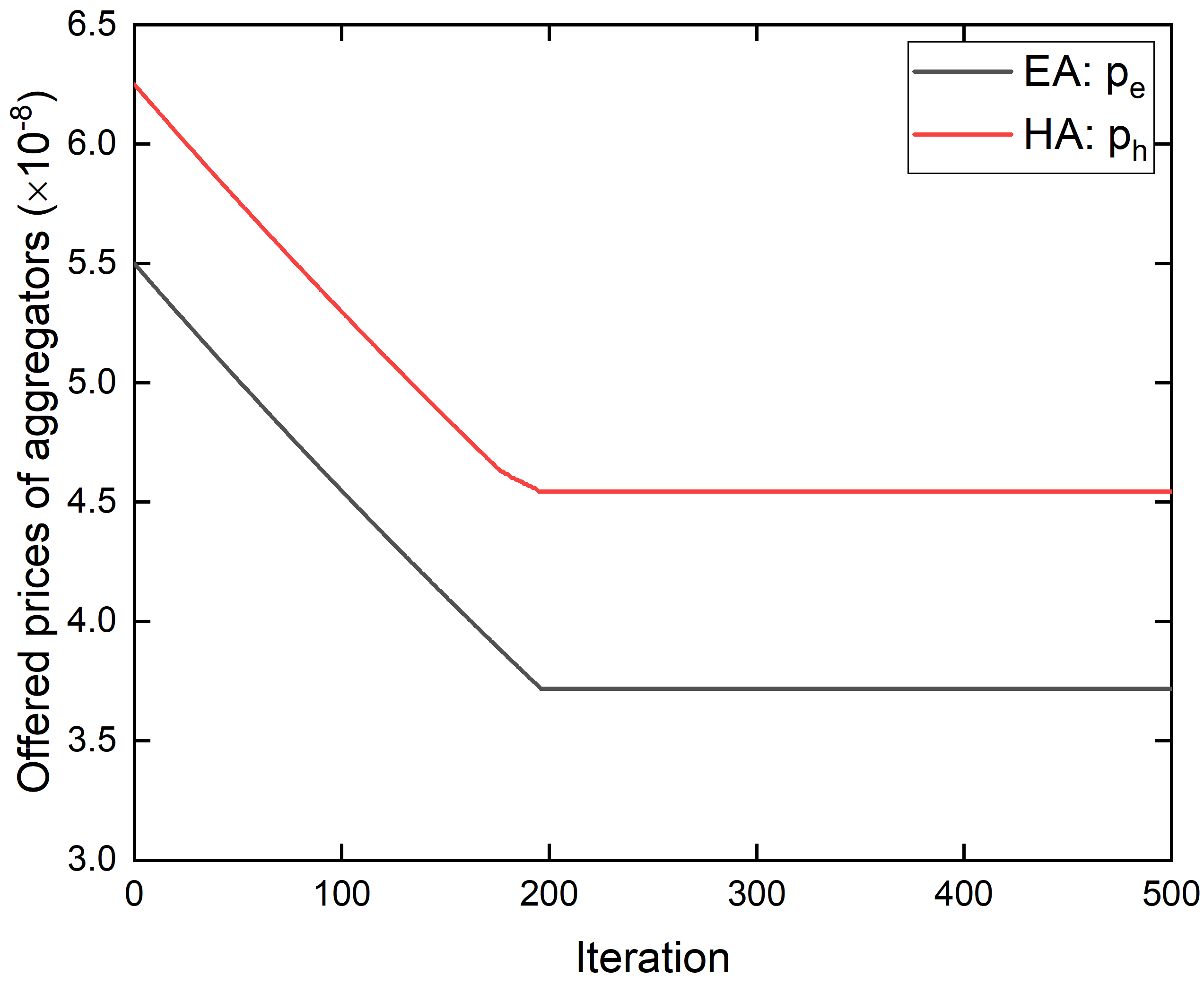}
	}%
	\subfigure[Initialization: $(r_e,r_h)$]{
		\includegraphics[width=0.48\linewidth]{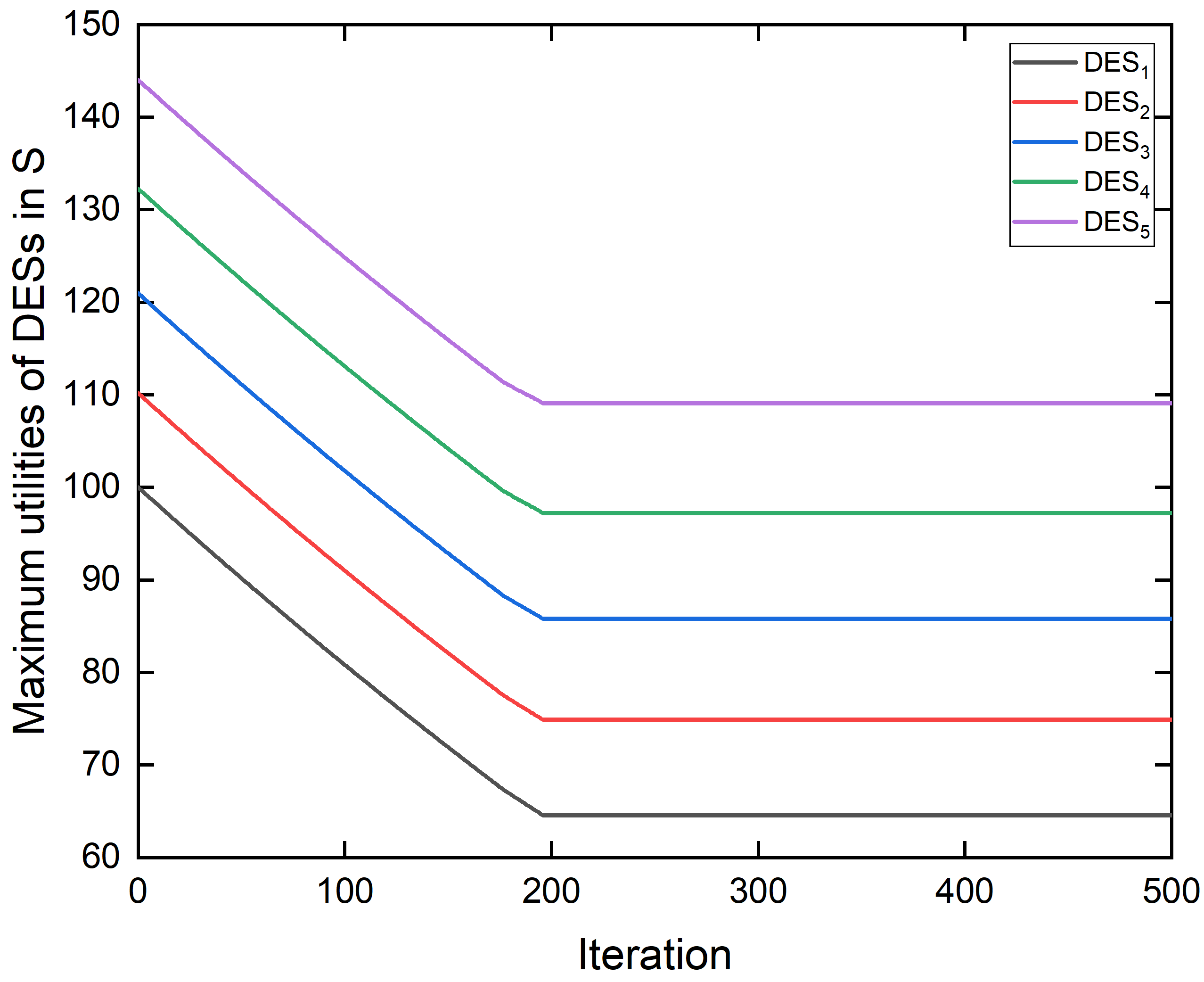}
	}%
	
	\subfigure[Initialization: $(c_e,c_h)$]{
		\centering
		\includegraphics[width=0.48\linewidth]{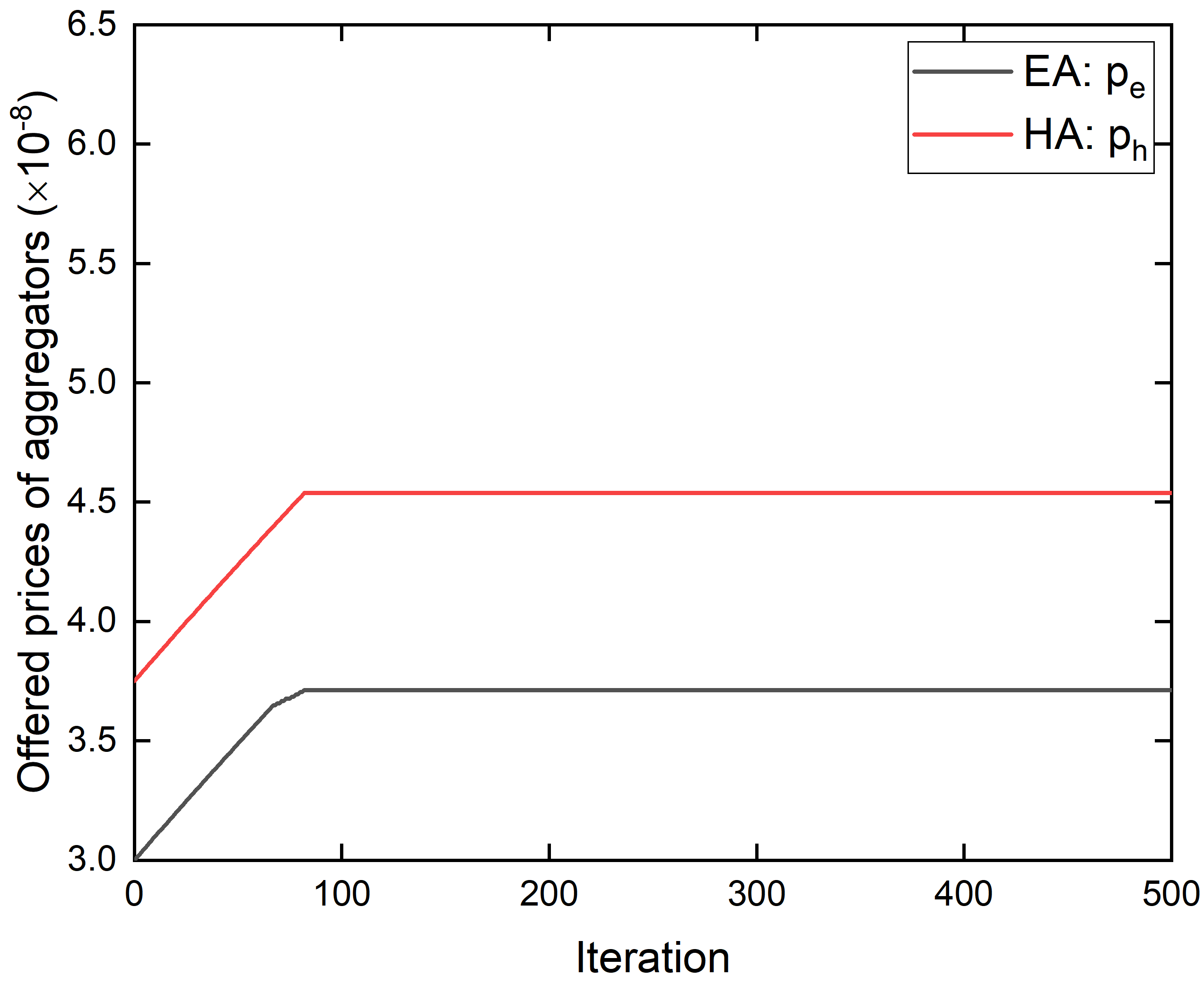}
	}%
	\subfigure[Initialization: $(c_e,c_h)$]{
		\includegraphics[width=0.48\linewidth]{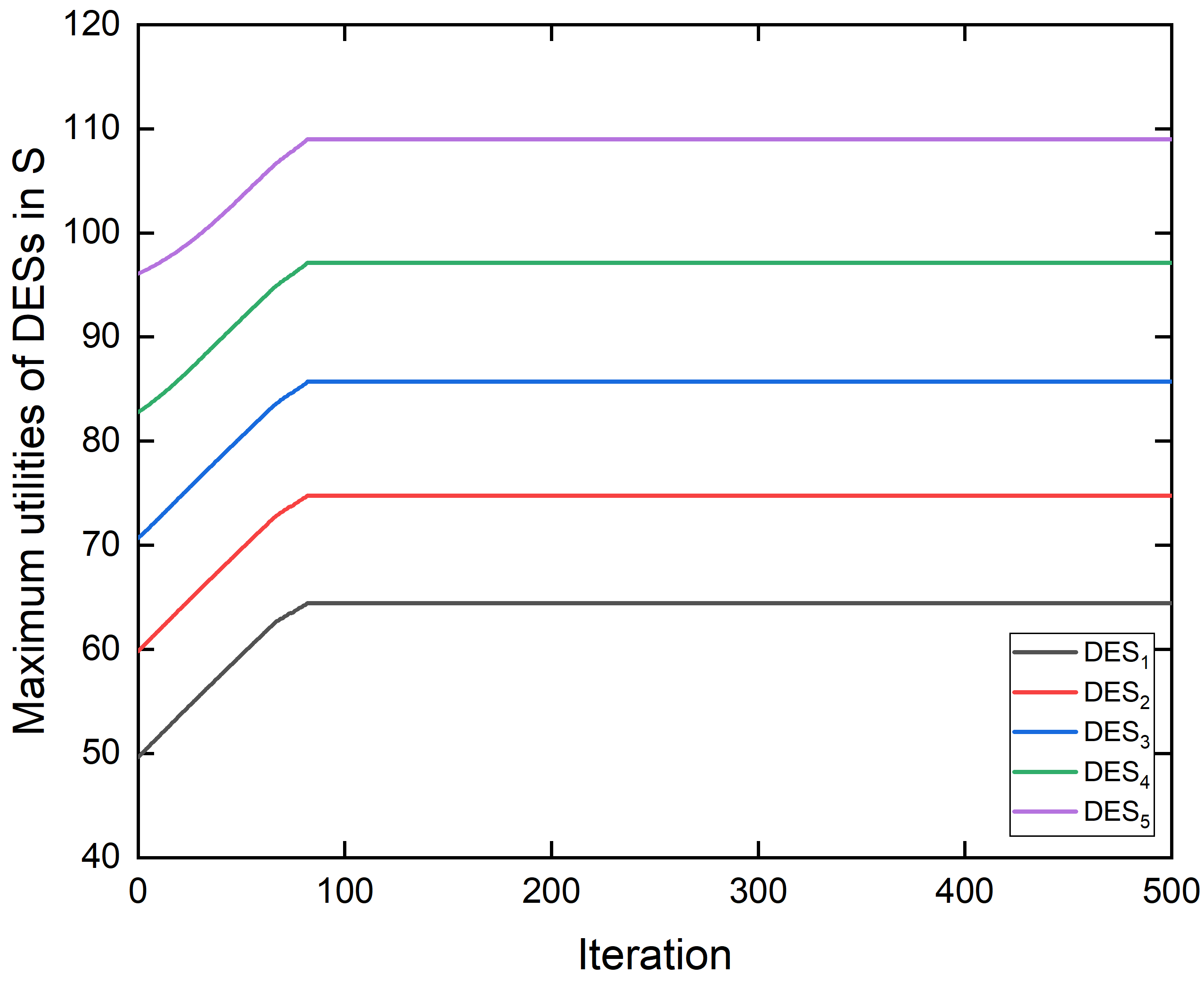}
	}%
	
	\subfigure[Initialization: $(\frac{c_e+r_e}{2},\frac{c_h+r_h}{2})$]{
		\centering
		\includegraphics[width=0.48\linewidth]{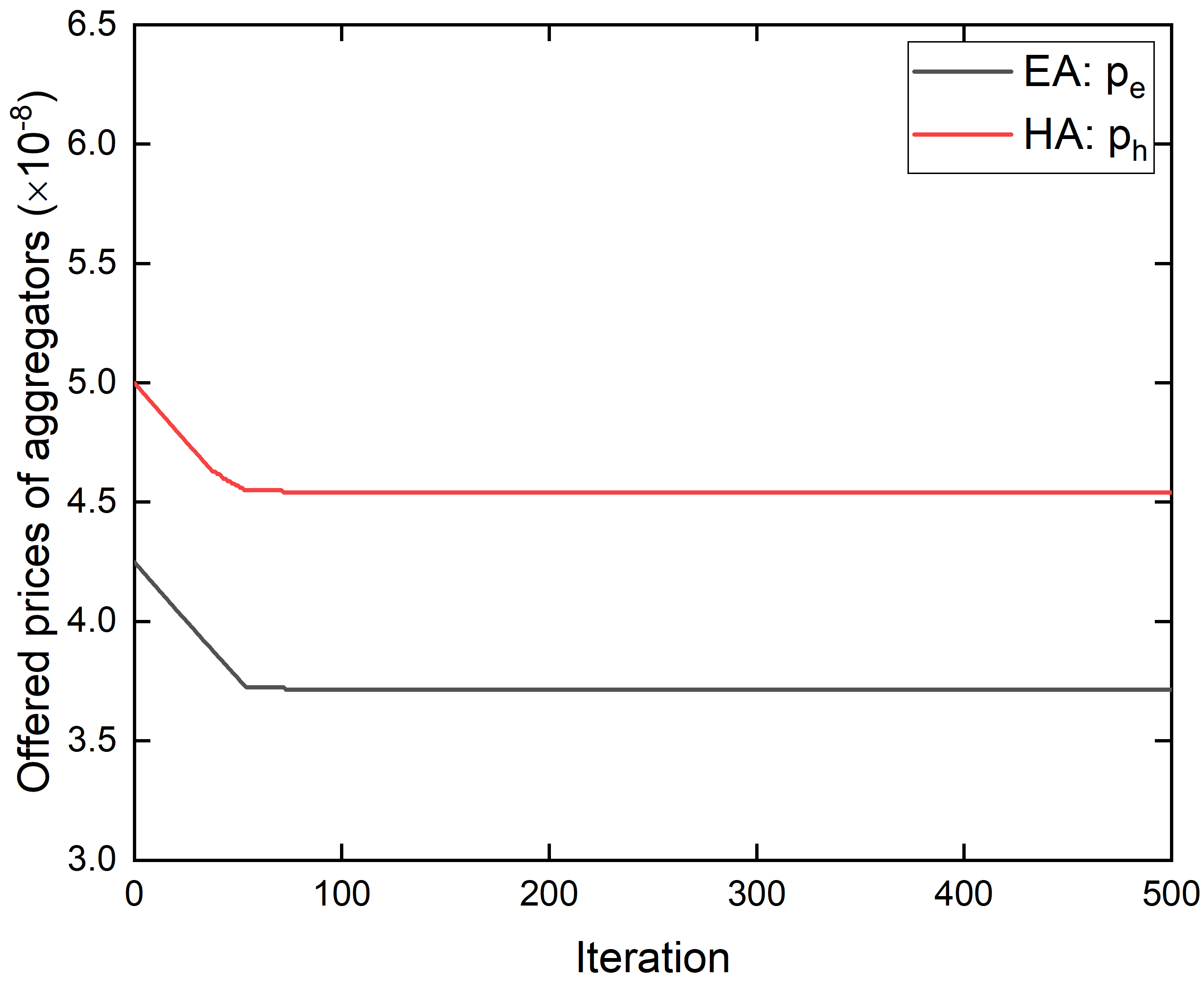}
	}%
	\subfigure[Initialization: $(\frac{c_e+r_e}{2},\frac{c_h+r_h}{2})$]{
		\includegraphics[width=0.48\linewidth]{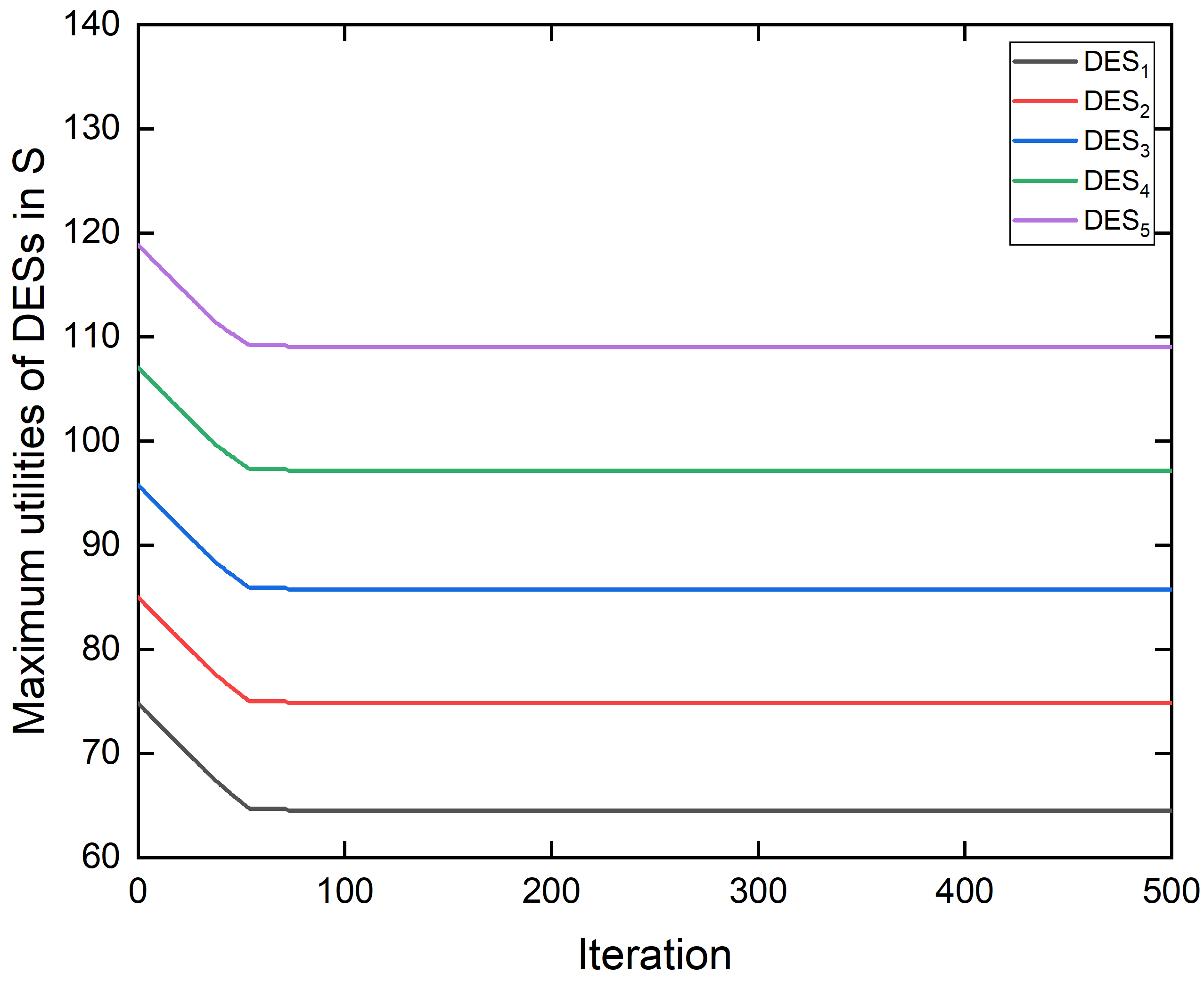}
	}%
	\centering
	\caption{The process of converging to Stackelberg equilibrium with different initializations under the restriction $M_1$.}
	\label{fig7}
\end{figure}

\textit{4) Stackelberg Equilibrium: } Consider a city ${\rm S}$ that has five communities, we define these ${\rm DES}_j$'s satisfaction coefficients as $k^j=(k_e^j,k_h^j)$ where $j\in\{1,\cdots,5\}$. We assume $k^1=(115.24, 137.81)$, $k^2=(129.14, 137.81)$, $k^3=(143.04, 137.81)$, $k^4=(156.94, 137.81)$, and $k^5=(170.85, 137.81)$ in this part. Fig. \ref{fig7} draws the process of converging to Stackelberg equilibrium with different initializations under the restriction $M_1$. Here, the parameters defined in Algorithm \ref{a1} is given by $\Delta=1\times10^{-10}$ and $\delta=0.999$. The initialization $(r_e,r_h)$ implies to give  $\{\tilde{p}_e^i,\tilde{p}_h^i\}\leftarrow\{r_e,r_h\}$ in line 3 of Algorithm \ref{a1}. Take (a) (b) in Fig. \ref{fig7} as an example, at the beginning, the aggregators offer the highest prices, thus they hardly gain any profit. By interacting with the five DES, the aggregators decrease their offering prices gradually in each iteration in order to improve profits. At approximately $100$-th iteration, they cannot improve their revenues by changing their strategies unilaterally, thus reaching the Nash Equilibrium. The DESs in ${\rm S}$ always respond aggregators with their optimal strategies, thus the Stackelberg equilibrium can be reached. From (a) (c) (e) in Fig. \ref{fig7}, we can see that they can reach the same equilibrium point regardless of what initialization is. However, the initialization affects the rate of convergence, and a good initialization can converge to the equilibrium point quickly. Fig. \ref{fig8} draws the process of converging to Stackelberg equilibrium with different $\Delta$ under the restriction $M_2$. Here, we adopt the initialization $(c_e, c_h)$ and $\delta=0.999$ as well. From (a) (c) in Fig. \ref{fig8}, they can quickly approach to equilibrium point when we adopt the larger $\Delta$. Nevertheless, it has to wait for $\Delta$ to drop to a relatively low level in order to improve this solution further. Therefore, how to choose the value of $\Delta$ depends on your demand. If we do not require high accuracy but high speed, it is recommended to choose a large $\Delta$; otherwise we should choose a small one.

\begin{figure}[!t]
	\centering
	\subfigure[$\Delta=1\times10^{-9}$]{
		\includegraphics[width=0.48\linewidth]{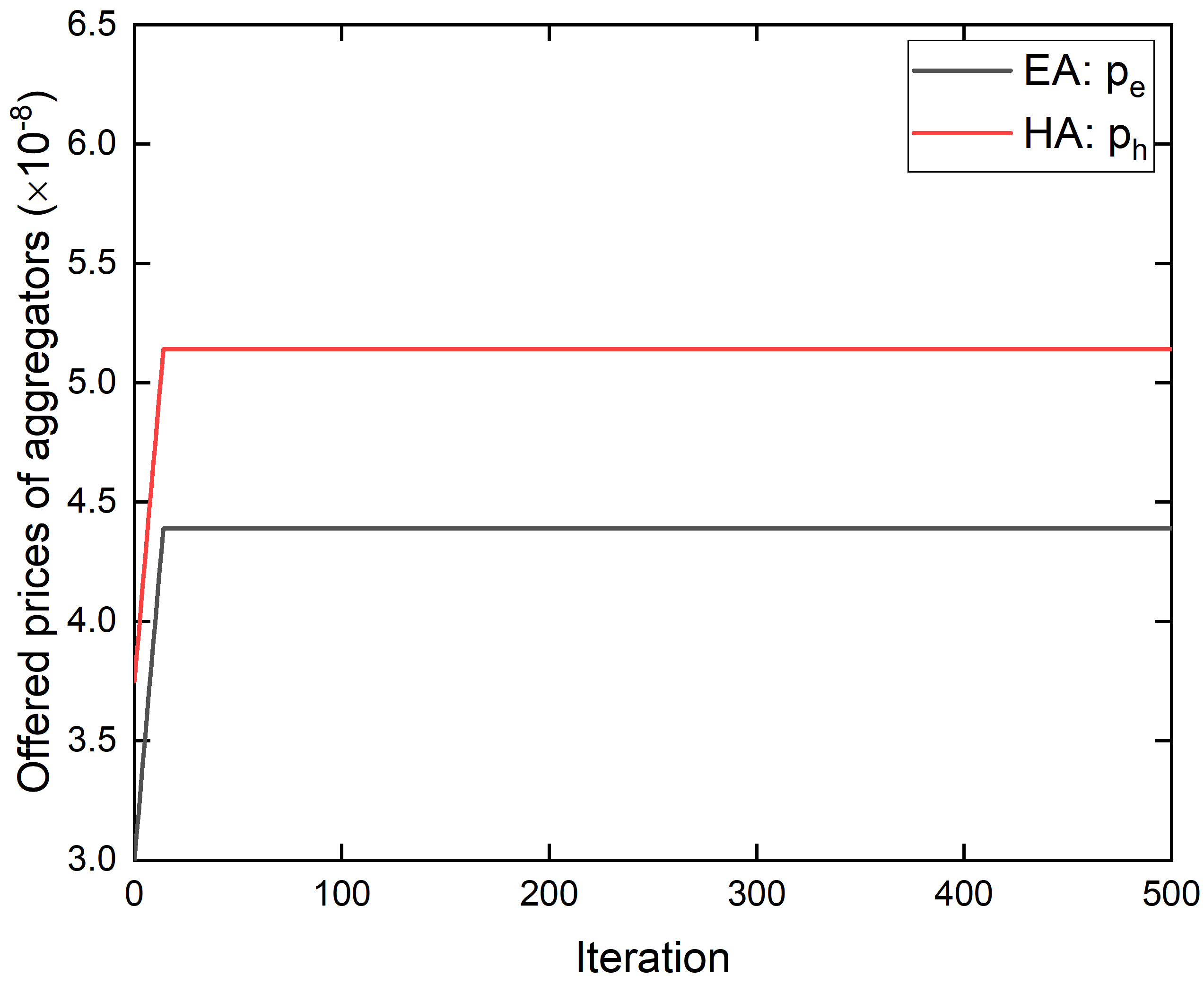}
	}%
	\subfigure[$\Delta=1\times10^{-9}$]{
		\includegraphics[width=0.48\linewidth]{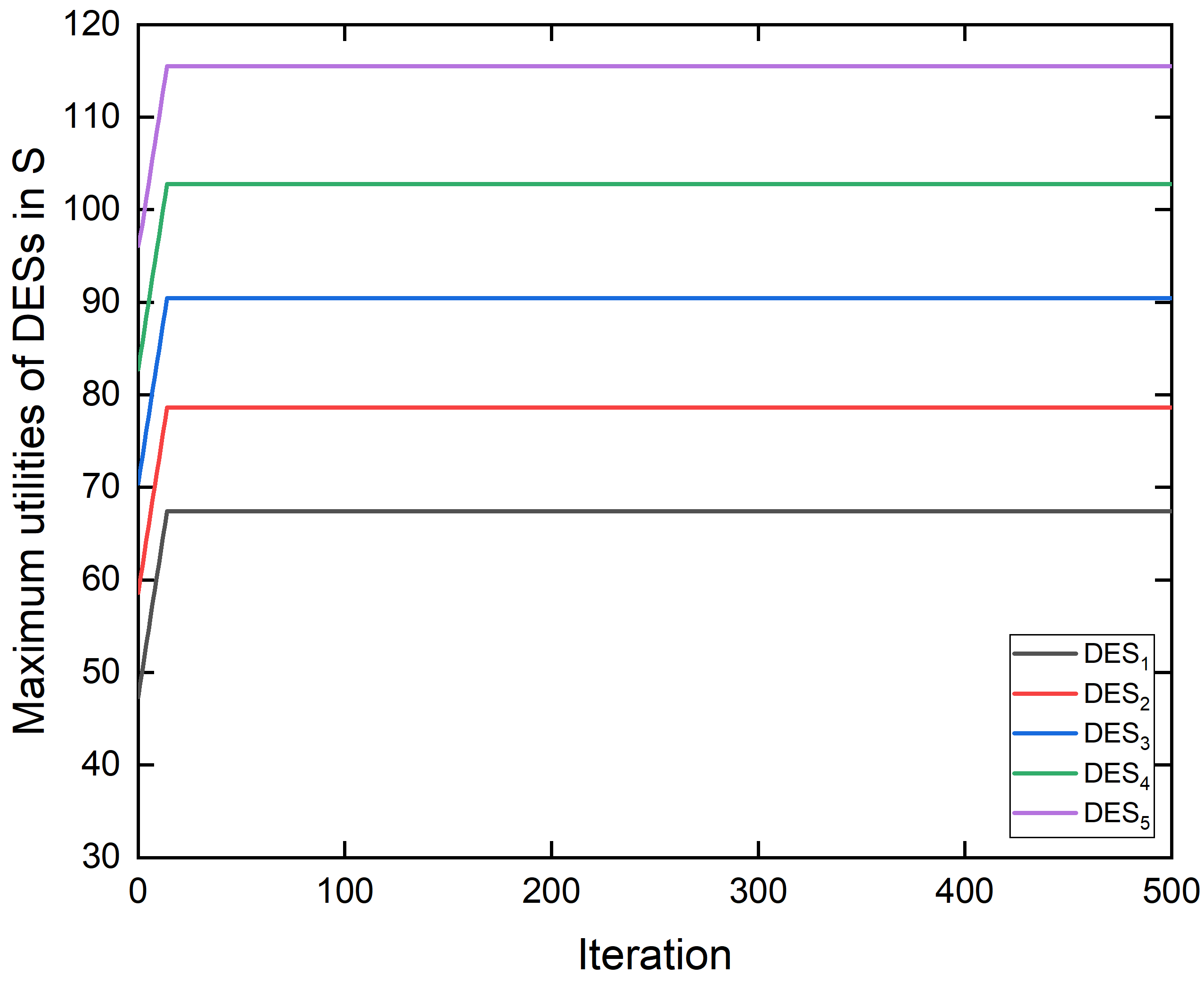}
	}%
	
	\subfigure[$\Delta=1\times10^{-10}$]{
		\centering
		\includegraphics[width=0.48\linewidth]{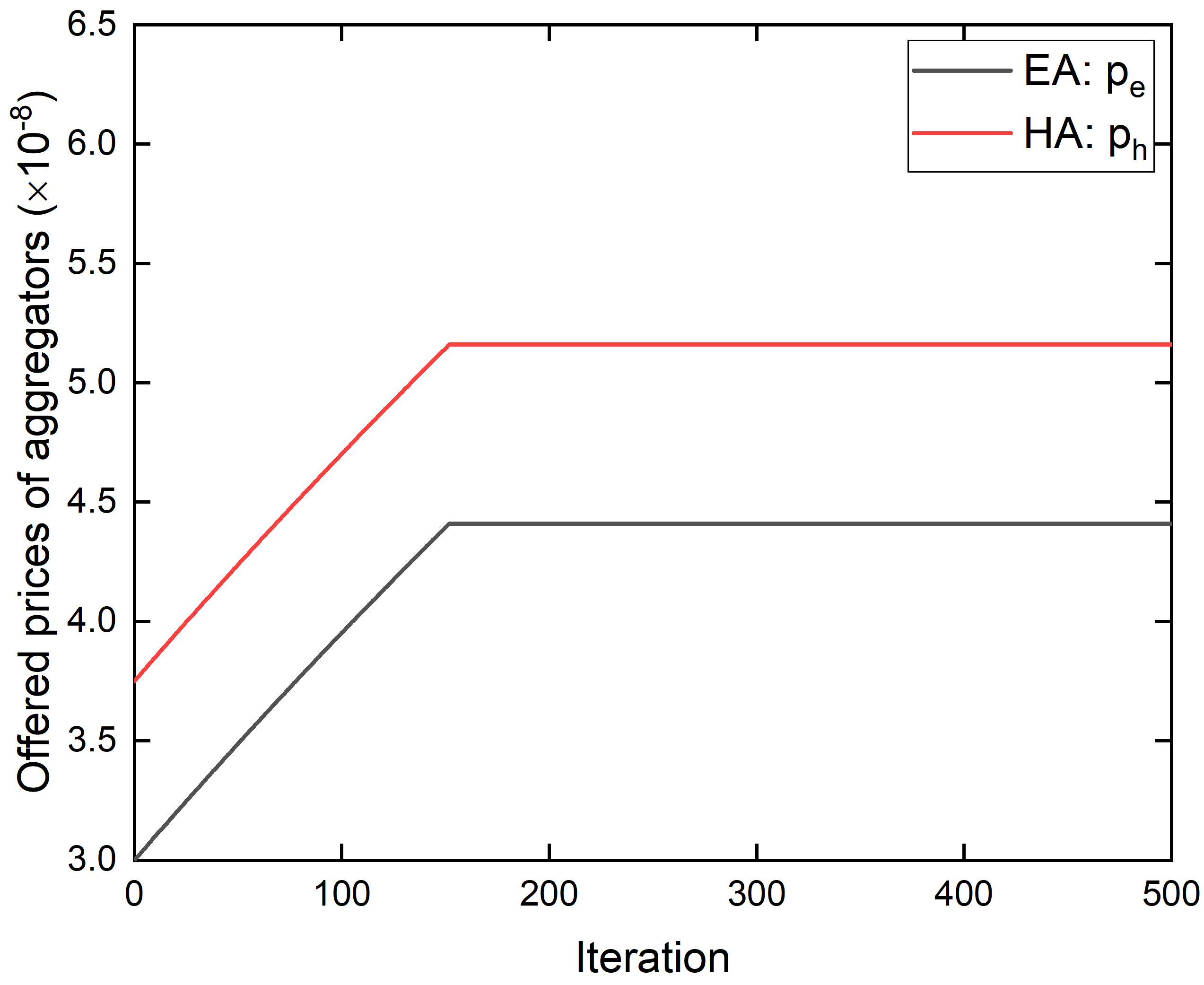}
	}%
	\subfigure[$\Delta=1\times10^{-10}$]{
		\includegraphics[width=0.48\linewidth]{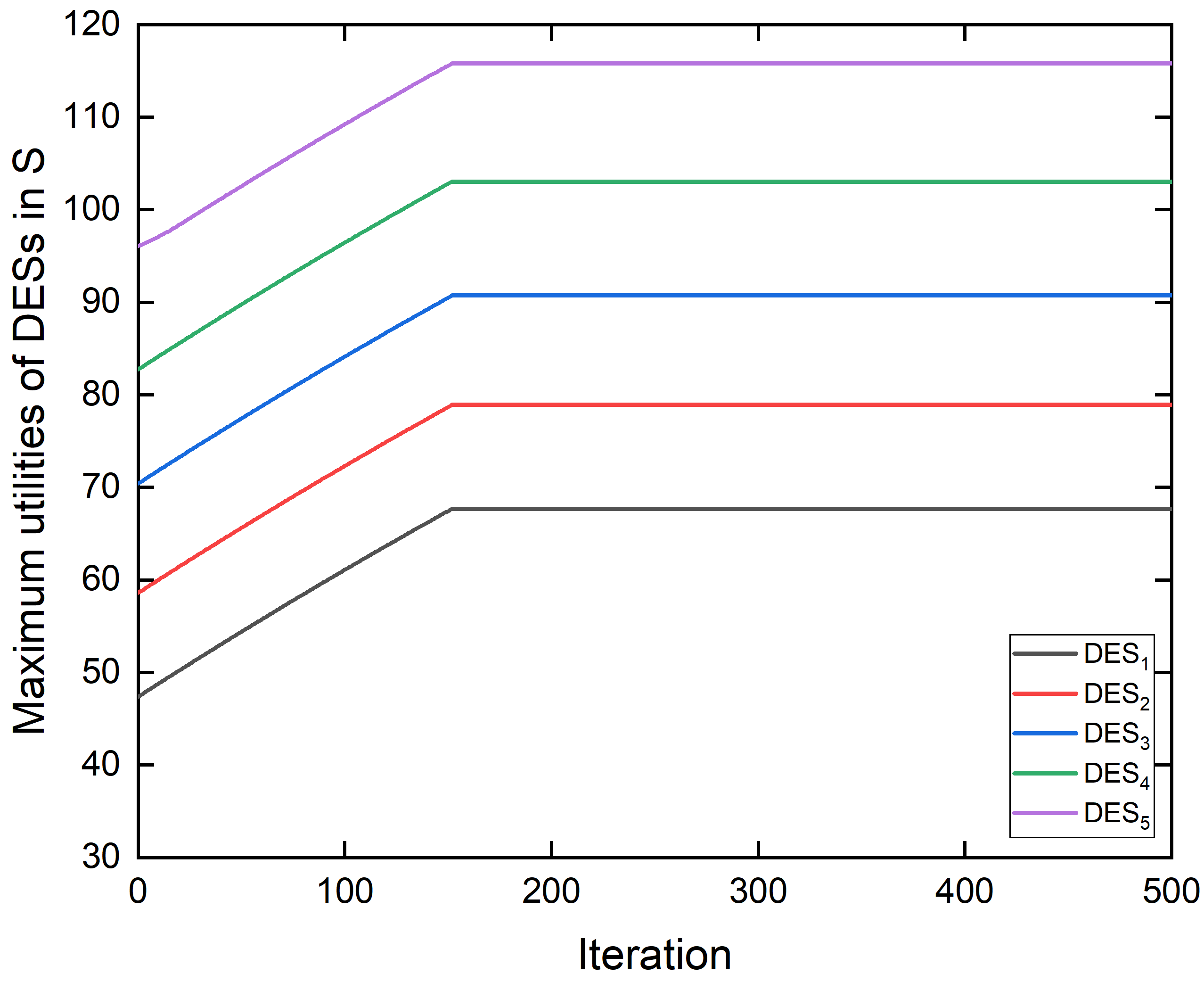}
	}%
	
	\centering
	\caption{The process of converging to Stackelberg equilibrium with different $\Delta$ under the restriction $M_2$.}
	\label{fig8}
\end{figure}

\textit{5) Centralized vs. Distributed: }For the aggregators, it is hard to know the complete information about all DESs in its city. Even if knowing partial coefficients, such as coefficient satisfactions, the settings of minimum energy restriction are very flexible, which will change with the fluctuations of the community population, season, and other factors. The optimal responses from DES are unpredictable. Thereby the aggregators can only obtain feedback information of DESs in a distributed manner, that is to update their offering price iteratively by interacting with DESs in their city.

\section{Conclusion}
In this paper, we studied multiple energies trading problem systematically. First, we proposed an architecture of B-MET system to address the security and privacy protection issues in distributed energy trading. In order to reduce latency and improve throughput, we introduce a credit model and design a new byzantine-based consensus mechanism based on it. Then, we model the interactions between aggregators and DESs in a smart city by MLMF Stackelberg game, which is more complex and realistic than the modes that have appeared before. We solve it step by step, show the existence and uniqueness of SE, and design a sub-gradient algorithm to find NE between aggregators. Finally, the results of numerical simulations indicated that our model is valid, and verify the correctness and efficiency of our algorithm.

\section*{Acknowledgment}

This work is partly supported by National Science Foundation under grant 1747818 and 1907472.

\ifCLASSOPTIONcaptionsoff
  \newpage
\fi

\bibliographystyle{IEEEtran}
\bibliography{references}

\end{document}